\documentclass[11pt,a4paper,fleqn]{article}

\usepackage[utf8]{luainputenc}

\usepackage{fullpage}

\usepackage{mathptmx}
\DeclareSymbolFont{symbols}{OMS}{cmsy}{m}{n} 

\usepackage{enumitem}
\setlist{nolistsep,leftmargin=\parindent}
\setitemize{label={--}}
\setenumerate{font=\footnotesize}

\usepackage{amssymb,
  amsmath,
  amsfonts,
  amsthm,
}

\theoremstyle{plain}
\newtheorem{theorem}{Theorem}[section]
\newtheorem{lemma}[theorem]{Lemma}
\newtheorem{corollary}[theorem]{Corollary}
\newtheorem*{claim*}{Claim}
\newtheorem{fact}[theorem]{Fact}
\newtheorem*{fact*}{Fact}
\newtheorem{proposition}[theorem]{Proposition}

\newtheorem{claim}{Claim}

\theoremstyle{definition}

\theoremstyle{remark}

\newcommand{\uend}{\hfill$\lrcorner$}

\newcommand{\definedas}{:=} 
\newcommand\NN{\mathbb{N}}

\renewcommand{\phi}{\varphi}

\newcommand\free{\operatorname{free}}
\newcommand\arity{\operatorname{ar}}
\newcommand\rank{\operatorname{qr}}

\newcommand\tw{\operatorname{tw}}
\newcommand\rootparam{\textsc{root}}
\newcommand\atomparam{\textsc{atom}}
\newcommand\cliqparam{\textsc{clique}}

\newcommand{\FO}{\textsc{fo}}
\newcommand{\CFO}{\textnormal{\textsc{cfo}}}
\newcommand{\OIFO}{\textnormal{{\footnotesize \textless}-inv-\textsc{fo}}}
\newcommand{\MSO}{\textsc{mso}}
\newcommand{\GSO}{\textsc{gso}}
\newcommand{\CMSO}{\textsc{cmso}}
\newcommand{\OIMSO}{\textnormal{{\footnotesize \textless}-inv-\textsc{mso}}}

\newcommand{\CC}{\mathcal{C}}

\newcommand{\formel}[1]{\textsf{\upshape #1}}
\newcommand{\oiatype}{\formel{oi-a-type}}
\newcommand{\oibtype}{\formel{oi-b-type}}

\newcommand{\atype}{\formel{a-type}}
\newcommand{\btype}{\formel{b-type}}

\newcommand{\dbstar}{{\star\star}}

\newcommand{\tp}{\textnormal{tp}}
\newcommand{\TP}{\textsc{tp}}

\newcommand{\bigmid}{\;\big|\;}

\renewcommand{\mathbf}[1]{\textit{\bfseries #1}}

\renewcommand{\bar}{\overline}

\newcommand{\angles}[1]{\left\langle#1\right\rangle}

\usepackage{tikz}

\usepackage{url} 
\usepackage[pdfborder={0 0 0}]{hyperref}

\usepackage{color}

\newcounter{rbcounter}
\newcommand{\randbem}[3]{\stepcounter{rbcounter}\parbox{0mm}{\hbox to 0pt{#1{}$^{\arabic{rbcounter}}$\hss}}\marginpar{#1\raggedright\scriptsize\textbf{#2$^{\arabic{rbcounter}}$:
    }#3}}

\begin{document}

\title{Order Invariance on Decomposable Structures}

\author{Michael~Elberfeld \and Marlin~Frickenschmidt \and Martin~Grohe}

\date{\textsc{rwth} Aachen University\\
  Aachen, Germany\\
  {\small \texttt{\{elberfeld|frickenschmidt|grohe\}@informatik.rwth-aachen.de}}\\[2ex]
  \today}

\maketitle

\begin{abstract} 
  Order-invariant formulas access an ordering on a structure's universe, but the
  model relation is independent of the used ordering. They are frequently used for
  logic-based approaches in computer science. Order-invariant formulas capture
  unordered problems of complexity classes and they model the independence of the
  answer to a database query from low-level aspects of databases. We study the
  expressive power of order-invariant monadic second-order (\MSO) and first-order
  (\FO) logic on restricted classes of structures that admit certain forms of tree
  decompositions (not necessarily of bounded width).

  While order-invariant \MSO\ is more expressive than \MSO\ and, even, \CMSO\
  (\MSO\ with modulo-counting predicates) in general, we show that order-invariant
  \MSO\ and \CMSO\ are equally expressive on graphs of bounded tree width and on
  planar graphs. This extends an earlier result for trees due to
  Courcelle. Moreover, we show that all properties definable in order-invariant
  \FO\ are also definable in \MSO\ on these classes. These results are
  applications of a theorem that shows how to lift up definability results for
  order-invariant logics from the bags of a graph's tree decomposition to the
  graph itself.

\end{abstract}

\section{Introduction}
\label{sec:introduction}

A formula is \emph{order-invariant} if it has access to an additional total
ordering on the universe of a given structure, but its answer is invariant with
respect to the given order. The concept of order invariance is used to formalize
the observation that logical structures are often encoded in a form that
implicitly depends on a linear order of the elements of the structure; think of
the adjacency-matrix representation of a graph. Yet the properties of structures
we are interested in should not depend on the encoding and hence the implicit
linear order, but just on the abstract structure. Thus, we use formulas that
access orderings, but define unordered properties. This approach can be
prominently found in database theory where formulas from first-order (\FO) and
monadic second-order (\MSO) logic are used to model query languages for
relational databases and (hierarchical) \textsc{xml} documents,
respectively. Being order-invariant means in this setting that the formula
evaluation process is always independent of low-level aspects of databases like,
for example, the encoding of elements as indices. Another example approach can
be found in descriptive complexity theory where formulas whose evaluation is
invariant with respect to specific encodings of the input structure capture
unordered problems decidable by certain complexity classes. The famous open
problem of whether there is a logic that captures all unordered properties
decidable in polynomial time falls into this category.

Gurevich~\cite{Gurevich1984} proved that order-invariant \FO\ (\OIFO) is more
expressive than \FO\ (also see \cite{Schweikardt2013} for details). The same
holds for order-invariant \MSO\ (\OIMSO) and \MSO\ with modulo-counting
predicates (\CMSO); Ganzow and Rubin showed that \OIMSO\ is able to express more
properties than \CMSO\ on general finite structures~\cite{GanzowR2008}. Since it
is not possible to decide, for a given \FO-formula, whether it is
order-invariant or not, this opens up the question of whether we can find
alternative logics that are equivalent to the order-invariant logics \OIFO\ and
\OIMSO. While on general logical structures no logics that are equivalent to
\OIFO\ or \OIMSO\ are known, this changes if we consider classes of structures
that are well-behaved. Benedikt and Segoufin~\cite{BenediktS2009} showed that
$\OIFO$ and $\FO$ have the same expressive power on the class of all strings and
the class of all trees (we write \emph{$\OIFO = \FO$ on $\CC$} to indicate that
the properties definable in $\OIFO$ equal the properties definable in $\FO$ when
considering structures from a class $\CC$). Considering $\OIMSO$,
Courcelle~\cite{Courcelle1996} showed that it has the same expressive power as
$\CMSO$ on the class of trees (that means, $\OIMSO = \CMSO$ on trees). Recently
it was shown that $\OIFO = \FO (= \MSO)$ and $\OIMSO = \CFO (= \CMSO) $ hold on
classes of graphs of bounded tree depth~\cite{EickmeyerEH2014}. More general
results that apply to graphs of bounded tree width or planar graphs have not
been obtained so far. This is due to the fact that, whenever we want to move
from an order-invariant logic to another logic on a class of structures, we need
to understand both (1) the expressive power of the order-invariant logic when
restricted to these structures, and (2) the ability of the other logic to handle
the structures in terms of, for example, definable decompositions.

\paragraph{Results.} Our results address both of these issues to better
understand the expressive power of order-invariant logics on decomposable
structures.

Addressing issue (1), we prove two general results, which show how to lift-up
definability results for order-invariant logics from the bags of tree
decompositions up to the whole decomposed structure. We show that, whenever we
are able to use \MSO-formulas to define a tree decomposition whose adhesion is
bounded (that means, bags have only bounded size intersections) and we can
define total orderings on the vertices of each bag individually, then $\OIMSO =
\CMSO$ (Theorem~\ref{th:lift-oimso}) and $\OIFO \subseteq \MSO$
(Theorem~\ref{th:lift-oifo}). Lifting theorems of this kind can be seen to be
implicitly used earlier~\cite{BenediktS2009,Courcelle1990,Courcelle1991}, but so
far they only applied to the case where the defined tree decomposition has a
bounded width. In this case, the whole structure can be easily transformed into
an equivalent tree. Our theorems also handle the case where bags have an
unbounded width: they merely assume the additional definability of a total
ordering on bags, possibly using arbitrary parameters (which may be sets in the
case of \MSO-definability). This is a much weaker assumption than having bounded
width, and it covers larger graph classes. The proofs of the lifting theorems
use type-composition methods to show how one can define the logical types of
structures from the logical types of substructures. The main challenge lies in
trading the power of the used types (in our case these are certain
order-invariant types based on orderings that are compatible with the given
decomposition) with the ability to prove the needed type-composition
methods. The latter need to work with bags of unbounded size and, thus, are more
general than the type-composition methods that are commonly used for the case of
bounded size bags.

Addressing issue (2), we study two types of classes of graphs where it is
possible to meet the assumptions of the lifting theorems and, thus, show that
$\OIMSO = \CMSO$ and $\OIFO \subseteq \MSO$ hold on these classes. The first two
results (formally stated as Theorems~\ref{th:tw-oimso} and~\ref{th:tw-oifo})
apply to classes of graphs of bounded tree width. For the proof, we show that
one can define tree decompositions of bounded adhesion in \MSO, where the bags
admit \MSO-definable total orderings. Let us remark that in proving these
results we do not rely on the \MSO-definability of width-bounded tree
decompositions, a result announced by Lapoire~\cite{lap98}, but only proved
recently (and independently of our work) by Boja\'nczyk and
Pilipczuk~\cite{BojanczykP2016a}~\cite{BojanczykP2016b}. Benedikt and Segoufin~\cite{BenediktS2009} had
shown earlier how to prove these results using the \MSO-definability of
width-bounded tree decompositions. Our second application of the lifting theorem
is concerned with classes of graphs that, for some $\ell\in\NN$, do not contain
$K_{3,\ell}$ as a minor. This includes the class of planar graphs and all
classes of graphs embedabble in a fixed surface \cite{rin65a,rin65b}. Using an
\MSO-definable tree decomposition into 3-connected components due to
Courcelle~\cite{Courcelle1999} along with proving that there are \MSO-definable
total orderings for the 3-connected bags of the decomposition, we are able to
apply the lifting theorems to prove that $\OIMSO = \CMSO$
(Theorem~\ref{th:minor-oimso}) and $\OIFO \subseteq \MSO$
(Theorem~\ref{th:minor-oifo}) hold on every class of graphs that exclude
$K_{3,\ell}$ as a minor for some $\ell \in \NN$.

\paragraph{Organization of the paper.} The paper starts with a preliminary section
(Section~\ref{sec:background}) containing definitions related to graphs and
logic. In Section~\ref{sec:lift}, we formally state and prove the lifting
theorems. Section~\ref{sec:decompositions} shows how to \MSO-define tree
decompositions along clique~separators and reviews the known \MSO-definable tree
decomposition into 3-connected components. Section~\ref{sec:orderings} picks up
the decomposed graphs and shows how to define total orderings for bags. This is
combined with the lifting theorems to prove the results about bounded tree width
graphs and $K_{3,\ell}$-minor-free graphs stated above. 

\section{Background}
\label{sec:background}

In the present section, we introduce the necessary background related to logical
structures and graphs (Section~\ref{sec:structures}), monadic second-order logic
and its variants (Section~\ref{sec:mso}), logical games and types
(Section~\ref{sec:games-types}), and transductions
(Section~\ref{sec:transductions}).

\subsection{Structures and Graphs}
\label{sec:structures}

A \emph{vocabulary} $\tau$ is a finite set of \emph{relational symbols} where an
\emph{arity} $\arity(R) \geq 1$ is assigned to each $R \in \tau$. A
\emph{structure} $A$ over a vocabulary $\tau$ consists of a finite set~$U(A)$,
its \emph{universe}, and a \emph{relation} $R(A) \subseteq U(A)^{\arity(R)}$ for
every $R \in \tau$. We sometimes write $R(A)$ by $R^A$, in particular if $R$ is
a symbol like $\le$.

An \emph{expansion} of a $\tau$-structure $A$ is a $\tau'$-structure $A'$ for
some vocabulary $\tau'\supseteq\tau$ such that $U(A) = U(A')$ and $R(A) = R(A')$
for all $R\in\tau$. If $A$ is a $\tau$-structure and $V\subseteq U(A)$, then the
\emph{induced substructure} $A[V]$ is the $\tau$-structure with universe
$U(A[V])=V$ and relations $R(A[V]):=R(A)\cap V^{\operatorname{ar}(R)}$ for all
$R\in\tau$. Furthermore, we let $A\setminus V:=A[U(A)\setminus V]$.

\emph{Graphs}~$G$ are structures over the vocabulary $\{E\}$ with $\arity(E) =
2$. When working with graphs, we also write $V(G)$ for the graph's universe (its
set of \emph{vertices}) and call $E(G)$ its set of \emph{edges}. The graphs we
are working with are \emph{undirected}. That means, for every two vertices $v$
and $w$, we have $(v,w) \in E(G)$ if, and only if, $(w,v) \in E(G)$ and
$(v,v)\not\in E(G)$. The \emph{Gaifman graph} $G(A)$ of a structure $A$ has
vertices $V(G(A)) = U(A)$ and for every pair of distinct elements $v$ and $w$
that are part of a common tuple in $A$, we insert the edge $(v,w)$ into
$E(G(A))$; thus, $G(A)$ is always undirected.

A \emph{tree decomposition} $(T,\beta)$ of a structure $A$ is a tree $T$
together with a labeling function $\beta \colon V(T) \to 2^{U(A)}$ satisfying
the following two conditions. (\emph{Connectedness condition}) For every element
$v \in U(A)$, the induced subtree $T \bigl[\{t \in V(T) \mid v \in
\beta(t)\}\bigr]$ is nonempty and connected. (\emph{Cover condition}) For every
tuple $(v_1,\dots,v_r)$ of a relation in $A$, there is a $t \in V(T)$ with
$\{v_1,\dots,v_r\} \subseteq \beta(t)$. It will be convenient to assume that the
trees underlying our tree decompositions are directed. That means, all edges are
directed away from a root. The set $N^T(t)$ of \emph{neighbors} of a node $t$
in a directed tree $T$ consists of its children (if $t$ is not a leaf) and its
parent (if $t$ is not the root). The set of children of a node $t$ in a directed
tree $T$ is denoted by $N^T_+(t)$. We omit $^T$ from $N^T(t)$ and $N^T_+(t)$ if
it is clear from the context. The sets $\beta(t)$ for every $t \in V(T)$ are the
\emph{bags} of the tree decomposition. The \emph{width} of the tree
decomposition is $\max_{t \in V(T)}\, |\beta(t)|-1$ and its \emph{adhesion} is
$\max_{(t,u) \in E(T)} |\beta(t) \cap \beta(u)|$. The \emph{tree width},
$\tw(A)$, of a structure $A$ is the minimum width of a tree decomposition for
it. Structures $A$ and their Gaifman graphs $G(A)$ have the same tree
decompositions. In particular $\tw(A) = \tw(G(A))$. The \emph{torso} of a node
$t \in V(T)$ in a tree decomposition $D = (T,\beta)$ for a structure $A$ with
Gaifman graph $G = G(A)$ is $G[\beta(t)]$ together with edges between all pairs
$v,w \in \beta(t) \cap \beta(u)$ for $u \in N(t)$.

\subsection{Monadic Second-Order Logic and its Variants}
\label{sec:mso}

\emph{Monadic second-order logic} (\MSO-logic) is defined by taking all
second-order formulas without second-order quantifiers of arity 2 and
higher. More specifically, to define its syntax, we use \emph{element variables}
$x_i$ for $i \in \NN$ and \emph{set variables} $X_i$ for $i \in
\NN$. \emph{Formulas of \MSO-logic} (\emph{\MSO-formulas}) over a vocabulary
$\tau$ are inductively defined as usual (see, for example,
\cite{Libkin2004}). Such formulas are also called $\MSO[\tau]$-formulas to
indicate the vocabulary along with the logic. The set of \emph{free variables}
of an $\MSO$-formula $\phi$, denoted by $\free(\phi)$, contains the variables of
$\phi$ that are not used as part of a quantification. By renaming a formula's
variables, we can always assume $\free(\phi) =
\{x_1,\dots,x_k,X_1,\dots,X_\ell\}$ for some $k,\ell \in \NN$; we write
$\phi(x_1,\dots,x_k,X_1,\dots,X_\ell)$ to indicate that the free variables of
$\phi$ are exactly $x_1$ to $x_k$ and $X_1$ to $X_\ell$. Given an \MSO-formula
$\phi(x_1,\dots,x_k,X_1,\dots,X_\ell)$, $A \models
\phi(a_1,\dots,a_k,A_1,\dots,A_\ell)$ indicates that $A$ together with the
assignment $x_i \mapsto a_i$, for $i \in \{1,\dots,k\}$, and $X_i \mapsto A_i$,
for $i \in \{1,\dots,\ell\}$, to $\phi$'s free variables satisfies $\phi$. A
formula without free variables is also called a \emph{sentence}.

\emph{Monadic second-order logic with modulo-counting} (\CMSO-logic) extends
\MSO-logic with the ability to access (built-in) \emph{modulo-counting atoms}
$C_{m}(R)$ for every $m \in \NN$ where $R$ is a relation symbol. Given a
structure $A$ over a vocabulary that contains $R$, we have $A \models C_{m}(R)$
exactly if $m$ divides $|R|$ (that means, $|R| \equiv 0 \mod m$). Atoms
$C_{m}(X)$ where~$X$ is a set variable are used in the same way.

Let $\tau$ be a vocabulary and $\le$ a binary relation symbol not contained in
$\tau$. An $\MSO$-sentence $\phi$ of vocabulary $\tau\cup\{\le\}$ is
\emph{order-invariant} if for all $\tau$-structures $A$ and all linear orders
$\le_1,\le_2$ of $U(A)$ we have $(A,\le_1)\models\phi$ if, and only if,
$(A,\le_2)\models\phi$. We can now form a new logic, \emph{order-invariant
monadic second-order logic} (\OIMSO-logic), where the sentences of vocabulary
$\tau$ are the order-invariant sentences of vocabulary $\tau\cup\{\le\}$, and a
$\tau$-structure $A$ satisfies an order-invariant sentence $\phi$ if $(A,\le)$
satisfies $\phi$ in the usual sense for some (and hence for all) linear orders
$\le$ of $U(A)$. There is a slight ambiguity in the definition of
order-invariant sentences in which binary relation symbol $\le$ we are referring
to as our special ``order symbol'' (there may be several binary relation symbols
in $\tau$). But we always assume that $\le$ is clear from the
context. Alternatively, we could view $\le$ as a ``built-in'' relation symbol
that is fixed once and for all and is not part of any vocabulary. However, this
would be inconvenient because we sometimes need to treat $\le$ just as an
ordinary relation symbol and the sentences of $\OIMSO$-logic of vocabulary
$\tau$ just as ordinary $\MSO$-sentences of vocabulary $\tau\cup\{\le\}$.

\emph{First-order logic} (\FO-logic) and \emph{order-invariant first-order
logic} (\OIFO-logic) are defined by taking all sentences of \MSO-logic and
\OIMSO-logic, respectively, that do not contain set variables.

\subsection{Games and Types}
\label{sec:games-types}

The \emph{quantifier rank} of an \MSO-formula $\phi$, denoted by $\rank(\phi)$,
is the maximum number of nested quantifiers in $\phi$. For structures $A,B$ and
$q\in\NN$, we write $A \equiv^\MSO_q B$ if $A$ and $B$ satisfy the same
$\MSO$-sentences of quantifier rank at most $q$. We write $A\equiv^{\OIMSO}_q B$
if $A$ and $B$ satisfy the same order-invariant $\MSO$-sentences of quantifier
rank at most $q$. For every $c\in\NN$, we write $A\equiv^{\CMSO}_{q,c} B$ if $A$
and $B$ satisfy the same $\CMSO$-sentences of quantifier rank at most $q$ and
only numbers $m \le c$ are used in the modulo-counting atoms.

It will sometimes be convenient to use versions of $\MSO$ and $\CMSO$ without
element variables (see, for example, \cite{Thomas1997}). In particular, in the
context of Ehrenfeucht-Fra{\"{\i}}ss{\'e} games. We will freely do so. We assume
that the reader is familiar with the characterizations of $\MSO$-equivalence and
$\CMSO$-equiva\-lence by Ehrenfeucht-Fra{\"{\i}}ss{\'e} games (see, for example,
\cite{ebbflu95,GanzowR2008}). Corresponding to the versions of the logics
without element variables, we use a version of the games where the players only
select sets and never elements, and a position \emph{induces a partial
isomorphism} if the mapping between the singleton sets of the position is a
partial isomorphism. (The rules of the game require the Duplicator to answer to
a singleton set with a singleton set and to preserve the subset relation.) Then
a \emph{position} of the game on structures $A,B$ is a sequence
$\Pi=(P_i,Q_i)_{i\in[p]}$ of pairs $(P_i,Q_i)$ of subsets $P_i\subseteq U(A)$
and $Q_i\subseteq U(B)$. The position is a \emph{$q$-move winning position} for
one of the players if this player has a winning strategy for the $q$-move game
starting in this position.

We also use the concept of \emph{types}. Let $\tau$ be a vocabulary and
$q,p\in\NN$. Then for all $\tau$-structures $A$ and sets
$P_1,\ldots,P_p\subseteq U(A)$, the \emph{\MSO-type of $(A,P_1,\ldots,P_p)$ of
quantifier rank $q$} is
\begin{align*}
  \tp_q^{\MSO}(A,P_1,\ldots,P_p)&:=\bigl\{ \phi(X_1,\ldots,X_p) \mid 
  \phi \text{ is \MSO-formula with } \rank(\phi) \leq q \text{ and } A \models\phi(P_1,\ldots,P_p)\}.
\end{align*}
Moreover, the class of all types over $\tau$ with respect to rank $q$ and $p$ free
set variables is 
\begin{align*}
  \TP^{\MSO}(\tau,q,p)&:=\big\{\tp_q^{\MSO}(A,P_1,\ldots,P_p) \mid A \text{ is
  $\tau$-structure}, P_1,\ldots,P_p\subseteq U(A)\big\}, 
\end{align*}
and we let $\TP^{\MSO}(\tau,q) := \TP^{\MSO}(\tau,q,0)$.  For $q,c\in\NN$, we
say that a $\CMSO$-formula has \emph{rank} at most $(q,c)$ if it has quantifier
rank at most $q$ and only contains modulo-counting atoms $C_m(X)$ with $m\le
c$. Based on this notion of rank, we define the $\CMSO$-type
$\tp_{q,c}^{\CMSO}(A,P_1,\ldots,P_p)$, and sets $\TP^{\CMSO}(\tau,q,c,p)$ and
$\TP^{\CMSO}(\tau,q,c)$.

Note that $\tp^\MSO_q(A,P_1,\ldots,P_p)=\tp_q^\MSO(B,Q_1,\ldots,Q_p)$ if, and
only if, $(P_i,Q_i)_{i\in[p]}$ is a $q$-move winning position for the Duplicator
in the $\MSO$-game on $A,B$. Furthermore, for $p=0$ we have $\tp_q(A)=\tp_q(B)$
if, and only if, $A\equiv^\MSO_q B$. Similar remarks apply to \CMSO-types.

For a vocabulary $\tau$ and a binary relation symbol $\le\, \notin \tau$, we say
that a subset $I\subseteq\TP^{\MSO}(\tau\cup\{\le\},q)$ is
\emph{order-invariant} if for all $\tau$-structures $A$ and all linear orders
$\le,\le'$ of $A$ we have $\tp_q^\MSO(A,\le)\in I$ if, and only if,
$\tp_q^\MSO(A,\le')\in I$. If $I$ is inclusion-wise minimal order-invariant,
then we call 
it an \emph{order-invariant type}. Note that every $\theta\in\TP^{\MSO}(\tau\cup\{\le\},q)$
is contained in exactly one order-invariant type, which we denote by
$\angles{\theta}$. We set
$\TP^{\OIMSO}(\tau,q):=\big\{\angles{\theta}\bigmid\theta\in\TP^{\MSO}(\tau\cup\{\le\},q)\big\}$,
the set of all order-invariant types. For a $\tau$-structure $A$, we call the
set $\tp_q^{\OIMSO}(A):=\angles{\tp_q^\MSO(A,\le)}$ for some and, hence, for all
linear orders of $A$ the \emph{order-invariant \MSO-type of $A$ of quantifier
rank $q$}. It may seem more natural to define the order-invariant type of a
structure as the set of all order-invariant sentences it satisfies. The
following proposition says that this would lead to an equivalent notion, but our
version is easier to work with, because it makes the connection between types of
ordered structures and order-invariant types more explicit.

\begin{lemma}
  \label{lem:oitypes}
  For all $\tau$-structure $A,A'$, the following statements are equivalent.
  \begin{enumerate}
    \item \label{item:types} $\tp_q^{\OIMSO}(A)=\tp_q^{\OIMSO}(A')$.
    \item \label{item:equiv} $A\equiv^\OIMSO_q A'$.
    \item \label{item:witness} There is a sequence $A_0,\ldots,A_\ell$ of $\tau$-structures
      and linear orders $\le_i,\le_i'$ with $A=A_0$, $A'=A_\ell$, and
      $(A_{i-1},\le_{i-1})\equiv_q^{\MSO}(A_i,\le_i')$ for all $i\in[\ell]$.
  \end{enumerate}
\end{lemma}

\noindent
If $A \equiv^\OIMSO_q A'$, we say that sequences
$(A_i)$, $(\le_i)$, and $(\le_i')$ as in statement \ref{item:witness} of
Lemma~\ref{lem:oitypes} \emph{witness} $A\equiv^\OIMSO_q A'$. 

\begin{proof}[Proof of Lemma~\ref{lem:oitypes}]
  We prove each of the implications from the chain
  \eqref{item:types}$\implies$\eqref{item:witness}$\implies$\eqref{item:equiv}$\implies$\eqref{item:types}.

  For proving \eqref{item:types}$\implies$\eqref{item:witness}, suppose
  $\tp_q^{\OIMSO}(A)=\tp_q^{\OIMSO}(A')$. Let $\theta:=\tp_q^\MSO(A,\le)$ for some
  linear order $\le$ of $A$ and $\theta':=\tp_q^{\OIMSO}(A',\le')$ for some linear
  order $\le'$ of $A'$.  Let $[A]$ be the class of all ordered $\tau$-structures
  $(A'',\le'')$ such that there is a sequence $A_0,\ldots,A_\ell$ of
  $\tau$-structures and linear orders $\le_i,\le_i'$ such that $A=A_0$ and
  $A''=A_\ell$ and $(A_{i-1},\le_{i-1})\equiv_q^{\MSO}(A_i,\le_i')$ for all
  $i\in[\ell]$, and let $[\theta]$ the class of types $\tp_q^\MSO(A'',\le'')$ for
  $(A'',\le'')\in[A]$. An easy induction on the length $\ell$ of the witnessing
  sequence shows that $[\theta]\subseteq\angles{\theta}$. Moreover, $[\theta]$ is
  order-invariant, and thus $[\theta]=\angles{\theta}$. Similarly, we define
  $[\theta']$ and prove that $[\theta']=\angles{\theta'}$. Thus
  $[\theta]=[\theta']$, and this implies \eqref{item:witness}.

  To prove \eqref{item:witness}$\implies$\eqref{item:equiv}, just note that all
  structures in a witnessing sequence satisfy the same order-invariant formulas.
 
  Finally, to prove \eqref{item:equiv}$\implies$\eqref{item:types}, suppose that
  $A\equiv^\OIMSO_q A'$. Let $\theta:=\tp^\MSO_q(A,\le)$ for some linear order $\le$ of
  $A$. Then $\tp_q^\OIMSO(A)=\angles{\theta}$. Let $\phi_{\angles{\theta}} :=
  \bigvee_{\theta'\in\angles{\theta}} \phi_{\theta'}$ with $\phi_{\theta'} :=
  \bigwedge_{\psi\in\theta'} \psi$.  Then $\phi_{\angles{\theta}}$ is an
  order-invariant \MSO-sentence of quantifier rank $q$. As $(A,\le)\models
  \phi_{\theta}$, we have $(A,\le)\models\phi_{\angles{\theta}}$, and thus $A$
  satisfies $\phi_{\angles{\theta}}$ as a sentence of $\OIMSO$. Hence $A'$
  satisfies $\phi_{\angles{\theta}}$ as a sentence of $\OIMSO$, and thus
  $(A',\le')\models \phi_{\angles{\theta}}$ for some linear order $\le'$ of
  $A'$. Thus there is a $\theta'\in\angles{\theta}$ such that $(A',\le')\models
  \phi_{\theta'}$, which implies $\tp_q^\MSO(A',\le')=\theta'$. Hence
  $\tp_q^\OIMSO(A') = \angles{\theta'}= \angles{\theta}$.
\end{proof}

\subsection{Transductions}
\label{sec:transductions}

Transductions define new structures out of a given structure. We use $w$-copying
\MSO-transductions as defined in \cite{CourcelleE2011}, but based on the below
terminology. They are able to (1) enlarge the universe of a given structure by
establishing $w$ copies of each element, (2) define relations over the new
universe from the given structure, and (3) not only define a single structure,
but a set of new structures parameterized by adding monadic relations to the
given structure.

An \emph{\MSO$[\tau,\tau']$-transduction of width $w$ with $p$ parameters} for
some $w,p \in \NN$ is defined via a finite collection $\Lambda$ of \MSO-formulas
over $\tau \cup \{P_1,\dots,P_p\}$ where the relation symbols $P_j$ are monadic
and not part of $\tau$. $\Lambda$ consists of a group of $w$ \MSO-formulas
$\lambda^1_U(x)$,\dots,$\lambda^w_U(x)$ for defining the universe of a new
structure and for each $R \in \tau'$ with some arity $r = \arity(R)$ a group of
$w^r$ formulas $\lambda^{(i_1,\dots,i_r)}_R(x_1,\dots,x_r)$ for $(i_1,\dots,i_r)
\in \{1,\dots,w\}^r$. Given a $\tau$-structure $A$ and $P_1,\dots,P_p \subseteq
U(A)$, they define the universe of a $\tau'$-structure
$\Lambda[A,P_1,\dots,P_p]$ via 
\begin{align*}
  U(\Lambda[A,P_1,\dots,P_p]) & := \{ (a,i) \in U(A) \times \{1,\dots,w\} \mid (A,P_1,\dots,P_p) \models \lambda^i_U(a) \} 
\end{align*}
and for each relation symbol $R \in \tau'$ the relation 
\begin{align*}
  R(\Lambda[A,P_1,\dots,P_p]) & := \{ ((a_1,i_1),\dots,(a_r,i_r)) \in (U(A)
                                \times \{1,\dots,w\})^r \mid A \models
                                \lambda^{(i_1,\dots,i_r)}_{R}(a_1,\dots,a_r)
                                \}\, .
\end{align*}
Finally, by ranging over all possible parameters, $\Lambda$ defines the set
\begin{align*}
  \Lambda[A] := \{ \Lambda[A,P_1,\dots,P_p] \mid P_1,\dots,P_p \subseteq U(A)
  \land (A,P_1,\dots,P_p) \models \lambda_{\textsc{valid}} \}
\end{align*}
for a given structure $A$ where $\lambda_{\textsc{valid}}$ is a formula that is
also part of the transduction, which singles out the valid combinations of the
given structure and parameters. Moreover, for a $\tau'$-structure $B$, we set
$\Lambda^{-1}[B] := \{ \tau\text{-structure } A \mid B \in \Lambda[A] \}$. For
an element $(a,i)$, we call $i$ its \emph{level}.

\MSO-transductions preserve \MSO-definability (formally stated by
Fact~\ref{fa:mso-closed}) and they can be composed to form new transductions
(formally stated by Fact~\ref{fa:trans-comp}). For a formal proof of
Fact~\ref{fa:trans-comp}, which implies Fact~\ref{fa:mso-closed}, see
\cite{CourcelleE2011}. The facts also hold if we replace all occurrences of
\MSO\ by \CMSO.

\begin{fact}[\MSO~is closed under \MSO-transductions]
  \label{fa:mso-closed}
  Let $\mathcal{P}$ be an \MSO-definable property of $\tau'$-structures and
  $\Lambda$ an \MSO$[\tau,\tau']$-transduction. Then the property of
  $\tau$-structures $\mathcal{P}' := \bigcup_{B \in \mathcal{P}}
  \Lambda^{-1}[B]$ is \MSO-definable.    
\end{fact}

\begin{fact}[\MSO-transductions are closed under composition]
  \label{fa:trans-comp}
  Let $\Lambda_1$ be an $\MSO[\tau,\tau']$-transduction and $\Lambda_2$ be an
  $\MSO[\tau',\tau'']$ for some vocabularies $\tau,\tau',\tau''$. Then there is
  an $\MSO[\tau,\tau'']$-transduction $\Lambda$ with $\Lambda[A] = \bigcup_{B
    \in \Lambda_1[A]} \Lambda_2[B]$ for every $\tau$-structure $A$. 
\end{fact}

\section{Lifting Definability}
\label{sec:lift}

An \emph{ordered tree decomposition} of a structure $A$ is a tree decomposition
of $A$ together with a linear order for each bag. We represent ordered tree
decompositions by logical structures in the following way. An \emph{ordered tree
extension} (\emph{otx} for short) of a $\tau$-structure $A$ is a structure
$A^\star$ that extends $A$ by a tree decomposition $(T^A,\beta^A)$ of $A$ and a
linear order $\preceq^A_t$ of $\beta^A(t)$ for each $t \in V(T^A)$. The
\emph{adhesion} of $A^\star$ is the adhesion of the tree decomposition
$(T^A,\beta^A)$. Formally, we view $A^\star$ as a structure over the vocabulary
$\tau^\star:=\tau\cup\{V_S,V_T,E_T,R_\beta,R_\preceq\}$, where $V_S$ and $V_T$
are unary, $E_T$ and $R_\beta$ are binary, and $R_\preceq$ is ternary. Of course
we assume that none of these symbols appears in $\tau$. In the
$\tau^\star$-structure $A^\star$, these symbols are interpreted as follows:
\begin{align*}
  V_S(A^\star)&:= U(A),\\
  V_T(A^\star)&:= V(T^A),\\
  E_T(A^\star)&:= E(T^A),\\
  R_\beta(A^\star)&:=\big\{(t,v)\mid t\in
                   V(T^A),v\in\beta^A(t)\big\}, \text{and}\\
  R_\preceq(A^\star)&:=\big\{(t,v,w)\mid t\in V(T^A) \text{ and } v,w\in\beta^A(t)\text{ with }v\preceq_t^Aw\big\}.
\end{align*}

An $\MSO[\tau,\tau^\star]$-transduction $\Lambda^\star$ \emph{defines an otx (of
adhesion at most $k$)} of a $\tau$-structure $A$ if every $B \in
\Lambda^\star(A)$ is isomorphic to an otx of $A$ (of adhesion
at most $k$) and $\Lambda^\star(A)$ is nonempty. We say that $\Lambda^\star$
\emph{defines otxs (of adhesion at most $k$)} on a class $\CC$ of
$\tau$-structures if $\Lambda^\star$ defines an otx (of adhesion at most $k$) of
every $A\in\CC$. Moreover, $\CC$ \emph{admits \MSO-definable ordered tree
decompositions (of bounded adhesion)} if there is such a transduction
$\Lambda^\star$ that defines otxs (of adhesion at most $k$ for some constant $k
\in \NN$) on $\CC$. We make similar definitions for the logic $\CMSO$.

We prove the following theorems, which show how to use the tree decompositions and
the bag orderings to define properties of order-invariant formulas without using
order invariance.

\begin{theorem}[Lifting theorem for \OIMSO]
  \label{th:lift-oimso} 
  Let $\CC$ be a class of structures that admits $\CMSO$-definable ordered
  tree decompositions of bounded adhesion. Then $ \OIMSO = \CMSO$ on $\CC$.
\end{theorem}

\begin{theorem}[Lifting theorem for \OIFO]
  \label{th:lift-oifo} 
  Let $\CC$ be a class of structures that admits $\MSO$-definable ordered
  tree decompositions of bounded adhesion. Then $\OIFO \subseteq \MSO$ on $\CC$.
\end{theorem}

Theorem~\ref{th:lift-oimso} is proved in three steps: First, in
Section~\ref{sec:segmented}, we modify the given ordered tree extension, such
that its tree decomposition follows a certain normal form that allows to
partition its nodes into two different classes (called a-nodes and b-nodes). The
partition of the nodes along with a global partial order that is based on the
local orderings in the bags is then encoded as part of the structure, turning
every otx into an expanded otx. Second, in Section~\ref{sec:ord-composition}, we
prove type-composition lemmas for both the a-nodes and the b-nodes. They show
how one can define the type of an expanded otx with respect to total orderings
that respect the already existing partial order from the types of substructures
that arise by adding such compatible orderings to them. Third,
Section~\ref{sec:inv-composition} shows how these type-composition lemmas can be
used in the context of order-invariance. Finally,
Section~\ref{sec:proofs-lifting} applies the type compositions to prove
Theorem~\ref{th:lift-oimso}. The proof of Theorem~\ref{th:lift-oifo} proceeds in
a similar way. The modifications that we need to apply to the proof of
Theorem~\ref{th:lift-oimso} in order to prove Theorem~\ref{th:lift-oifo} are
mentioned along the way.
 
\subsection{Segmented Ordered Tree Extensions}
\label{sec:segmented}

Recall that we view the tree in a tree decomposition as directed. A tree
decomposition $(T,\beta)$ of a structure $A$ is \emph{segmented} if the set
$V(T)$ can be partitioned into a set $V_a$ of \emph{adhesion nodes} and a set
$V_b$ of \emph{bag nodes} (\emph{a-nodes} and \emph{b-nodes}, for short)
satisfying the following conditions.
\begin{enumerate}
\item For all edges $tu\in E(T)$, either $t\in V_a$ and $u\in V_b$ or
  $u\in V_a$ and $t\in V_b$.
\item For all a-nodes $t\in V_a$ and all
  distinct neighbors $u_1,u_2\in N(t)$, we have
  $\beta(t)=\beta(u_1)\cap\beta(u_2)$.
\item For all b-nodes $t\in V_b$ and all distinct neighbors $u_1,u_2\in N(t)$ we have
  $\beta(t)\cap\beta(u_1)\neq\beta(t)\cap\beta(u_2)$.
\item All leaves of $T$ are b-nodes.
\end{enumerate}
We can transform an arbitrary tree decomposition $(T,\beta)$ into a segmented
tree decomposition $(T'',\beta'')$ as follows. In the construction, we view $T$
as an undirected tree. We will have $V(T)\subseteq V(T'')$. Thus we can direct
the edges of $T''$ away from the root of $T$, which will remain the root of
$T''$. We first contract all edges $tu \in E(T)$ with
$\beta(u)\subseteq\beta(t)$, resulting in a decomposition $(T',\beta')$ where
$\beta'(u)\not\subseteq\beta'(t)$ for all $tu\in E(T')$. Then, for all edges
$tu\in E(T')$, we introduce a new node $v_{tu}$, where $v_{tu}=v_{ut}$, and
edges from $v_{tu}$ to $t$ and $u$. Then we identify all nodes $v_{tu}$ and
$v_{tu'}$ such that $\beta'(t)\cap\beta'(u)=\beta'(t)\cap\beta'(u')$. We let
$T''$ be the resulting tree. The nodes from the original tree $T$ are the
b-nodes, and the nodes $v_{tu}$ are the a-nodes. We define $\beta''$ on $V(T'')$
by $\beta''(t):=\beta'(t)$ for $t\in V(T')$ and
$\beta''(v_{tu}):=\beta'(t)\cap\beta'(u)$ for all $tu\in E(T')$. The resulting
tree decomposition $(T'',\beta'')$ is segmented. This transformation is
definable by an \MSO-transduction. Thus we may assume that the tree
decompositions in ordered tree extensions are segmented, because there is an
$\MSO[\tau^\star,\tau^\star]$-transduction $\Lambda_{\textsc{segment}}$ that
transforms every otx into an otx where the tree decomposition is segmented.

For the rest of this section, we fix a vocabulary $\tau$ that does not contain
the order symbol $\le$ and a $k\in\NN$. In the rest of this section, we only
consider otxs of $\tau$-structures. We assume that the adhesion of these otxs is
at most $k$ and their tree decomposition is segmented.

It will be convenient to introduce some additional notation. As before, whenever
we denote an otx by $A^\star$, we denote the underlying structure by $A$ and the
tree decomposition by $(T^A,\beta^A)$.  We denote the descendant order in the
tree $T^A$ of an otx $A^\star$ by $\unlhd^A$. For every node $t\in V(T^A)$, we
let $T^A_t$ be the subtree of $T^A$ rooted in $t$, that is, $T^A_t:=T^A[\{u\in
V(T^A)\mid t\unlhd^A u\}]$. We let $\gamma^A(t)$, called the \emph{cone} of $t$,
be the union of all bags $\beta^A(u)$ for $u\in V(T^A_t)$. If $s$ is the parent
of $t$ we let $\sigma^A(t):=\beta^A(t)\cap\beta^A(s)$; this is the
\emph{separator at} $t$. For the root $r$ we let $\sigma^A(r):=\emptyset$.  In
all these notations we may omit the index ${}^A$ if $A$ is clear from the
context. Note that for all a-nodes $t$ of $T$ and all $u\in N_+(t)$ we have
$\sigma(t)=\beta(t)=\sigma(u)$.

We expand an otx $A^\star$ to a structure $A^\dbstar$ over the vocabulary
$\tau^\dbstar:=\tau^\star\cup\{V_a,V_b,R_\gamma,R_\sigma,S_1,\ldots,S_k,\preceq\}$, 
where $V_a,V_b$ are unary and $R_\sigma,R_\gamma,S_1,\ldots,S_k,\preceq$ are
binary relation symbols that do not appear in $\tau$. We let $V_a(A^\dbstar)$
and $V_b(A^\dbstar)$ be the sets of a-nodes and b-nodes of the tree $T^A$,
respectively, and
\begin{align*}
  R_\sigma(A^\dbstar)&:=\big\{(t,v)\mid t\in 
                   V(T^A),v\in\sigma^A(t)\big\},\\
  R_\gamma(A^\dbstar)&:=\big\{(t,v)\mid t\in 
                   V(T^A),v\in\gamma^A(t)\big\}.
\end{align*}
We let $\preceq = \preceq^{A^\dbstar}$ be the partial order on $U(A^\dbstar)$
defined as follows. We first define the restriction of $\preceq$ to $V(T)$. For
all b-nodes $t$, we let $\preceq'_t$ be the linear order on $N_+(t)$ defined by
$u_1\preceq_t' u_2$ if the set $\sigma(u_1)\subseteq\beta(t)$ is
lexicographically smaller than or equal to the set
$\sigma(u_2)\subseteq\beta(t)$ with respect to the linear order $\preceq_t$ on
$\beta(t)$, for all children $u_1,u_2\in N_+(t)$. This is indeed a linear order
because $\preceq_t$ is a linear order of $\beta(t)$ and
$\sigma(u_1)\neq\sigma(u_2)$ for all distinct $u_1,u_2\in N_+(t)$.  Then we let
the restriction of $\preceq$ to $V(T)$ be the reflexive transitive closure of
the ``descendant order'' $\unlhd$ on $T$ and all the relations $\preceq_t'$ for
b-nodes $t\in V(T)$. To define the restriction of $\preceq$ to $U(A)$, for every
$v\in U(A)$ we let $t(v)$ be the topmost (that is, $\unlhd$-minimal) node $t\in
V(T)$ such that $v\in\beta(t)$. Then we let $v\preceq w$ if, and only if,
$t(v)\prec t(w)$ or $t(v)=t(w)$ and $v\preceq_{t(v)} w$. To complete the
definition of $\preceq$, we let $t\preceq v$ for all $t\in V(T)$ and $v\in
U(A)$.

Finally, we define the relations $S_1(A^\dbstar),\ldots,S_k(A^\dbstar)$ by
letting $S_i(A^\dbstar)$ be the set of all pairs $(t,v)$, where $t\in V(T^A)$
and $v$ is the $i$th element of $\sigma(t)$ with respect to the partial order
$\preceq$, which is a linear order when restricted to
$\sigma(t)\subseteq\beta(t)$. Recall that we have $|\sigma(t)|\le k$ by our
general assumption that the adhesion of all otxs is at most $k$. This completes
the definition of $A^\dbstar$. It is easy to see that there is an
$\MSO[\tau^\star,\tau^{\dbstar}]$-transduction $\Lambda_{\textsc{expand}}$ that
defines $A^\dbstar$ in $A^\star$.

We call $A^\dbstar$ an \emph{expanded otx} (\emph{otxx} for short) of $A$. More
generally, we call a $\tau^\dbstar$-structure $A'$ an \emph{expanded otx} if
there is a $\tau$-structure $A$ such that $A'$ is an otxx of $A$.  Let
$A^\dbstar$ be an expanded otx. For every $t\in V(T)$, we let
\begin{align*}
  A_t^\dbstar & := A^\dbstar[\gamma(t)\cup V(T_t)], \text{ and}\\
  A_{(t)}^\dbstar & := A^\dbstar[\beta(t)\cup N_+(t)] .
\end{align*}
We call a $\tau^\dbstar$-structure $A'$ a \emph{sub-otxx} if there is an otxx
$A^\dbstar$ and a node $t\in V(T^A)$ with $A'=A_t^\dbstar$. The only difference
between an otxx and a sub-otxx is that in an otxx the set $\sigma(r)$ is empty
for the root $r$ whereas in a sub-otxx it may be nonempty.

\begin{lemma}\label{lem:def-otxx}
  There are \MSO-sentences $\formel{otxxs}$ and $\formel{sub-otxx}$
  of vocabulary $\tau^\dbstar$ defining the classes of all otxx and
  sub-otxx (satisfying our general assumptions: the tree decomposition
  is segmented and has adhesion at most $k$).
\end{lemma}

\begin{proof}
  Straightforward.
\end{proof}

We will later modify an otxx $A^\dbstar$ by
\emph{replacing}\label{page:replacing} a sub-otxx $A^\dbstar_{t}$, for some
$t\in V(T^A)$, by another sub-otxx $B^\dbstar$. Let $t'$ be the root node of the
tree $T^B$. The replacement is possible if the induced substructures
$A^\dbstar[\{t\}\cup\sigma^A(t)]$ and $B^\dbstar[\{t'\}\cup\sigma^B(t')]$ are
isomorphic. If they are, there is a unique isomorphism, because
$\{t\}\cup\sigma^A(t)$ and $\{t'\}\cup\sigma^B(t')$ are linearly ordered by the
restrictions of $\preceq^{A^\dbstar}$, $\preceq^{B^\dbstar}$. Now replacing
$A^\dbstar_{t}$ by $B^\dbstar$ in $A^\dbstar$ just means deleting all elements
in $U(A^\dbstar_t)$ except those in $\{t\}\cup\sigma^A(t)$, adding a disjoint
copy of $B^{\dbstar}$, and identifying the elements in $\{t\}\cup\sigma^A(t)$
and $\{t'\}\cup\sigma^B(t')$ according to the unique isomorphism.  Note that the
substructures $A^\dbstar[\{t\}\cup\sigma^A(t)]$ and
$B^\dbstar[\{t'\}\cup\sigma^B(t')]$ are isomorphic if the sub-otxxs
$A_t^\dbstar$ and $B^\dbstar$ satisfy the same first-order sentences of
quantifier rank $\operatorname{ar}(\tau)+1$, where $\operatorname{ar}(\tau)$
denote the maximum arity of a relation symbol in the vocabulary $\tau$. To
express isomorphism, we use the relations $S_1,\ldots,S_k$ and the fact that the
root of an otxx can be defined by a formula of quantifier rank $2$. Thus in
particular, if $\tp_q^{\MSO}(A_t^\dbstar)=\tp_q^{\MSO}(B^\dbstar)$ for some
$q\ge \operatorname{ar}(\tau)+1$, we can replace $A^\dbstar_t$ by $B^\dbstar$.

Finally, we say that a linear order $\le$ on an otxx or sub-otxx $A^\dbstar$ is
\emph{compatible} if it extends the partial order $\preceq^{A^\dbstar}$. If
$\le$ is a compatible linear order, then $(A^\dbstar,\le)$ denotes the
$\tau^\dbstar\cup\{\le\}$-expansion of $A^\dbstar$ by this order, and
$(A_t^\dbstar,\le)$ denotes the induced substructure where $\le$ is restricted
to the sub-otxx $A_t^\dbstar$. We can extend the replacement operation to such
ordered expansions of otxxs; in the same way we replace a sub-otxx $A^\dbstar_t$
by $B^\dbstar$, we can replace a $(A^\dbstar_t,\le)$ by $(B^\dbstar,\le')$ for
some compatible linear order $\le'$ of $B^\dbstar$.

\subsection{Ordered Type Compositions}
\label{sec:ord-composition}

As all structures we are working with in this subsection are otxxs and sub-otxx,
we denote them by $A$ rather than $A^\dbstar$. Apart from that, we use the same
notation as before. In particular, if $A$ is an otxx then by $T^A$ we denote the
tree of its tree decomposition, and for a node $t\in V(T^A)$, by $A_t$ we denote
the sub-otxx rooted in $t$, and we let $A_{(t)} = A[\beta(t)\cup N_+(t)]$.

Throughout this subsection, we fix a $q\in\NN$ such that $q\ge 2$ and
$q\ge\operatorname{ar}(\tau)+1$ and $q$ is at least the quantifier rank of the
formulas $\formel{otxx}$ and $\formel{sub-otxx}$ of
Lemma~\ref{lem:def-otxx}. This means that if $A$ is an otxx (or sub-otxx) and
$A'$ an arbitrary $\tau^\dbstar$-structure with $A \equiv^\MSO_q A'$, then $A'$
is an otxx (a sub-otxx) as well. Furthermore, if $t,t'$ are the root nodes of
$A$, $A'$, respectively, then the induced substructures
$A[\{t\}\cup\sigma^A(t)]$ and $A'[\{t'\}\cup\sigma^{A'}(t')]$ are
isomorphic. Finally, if $A,A'$ are otxxs and $\le,\le'$ are linear orders of
$A,A'$, respectively, such that $(A,\le)\equiv_q^\MSO(A',\le')$ then $\le$ is
compatible if, and only if, $\le'$ is compatible.

We let $\Theta:=\TP^{\MSO}(\tau^\dbstar\cup\{\le\},q)$. Furthermore, we assume
that $\Theta=\{\theta_1,\ldots,\theta_m\}$.

Let $A$ be an otxx, $\le$ a compatible linear order of $A$, and $N\subseteq
V(T^A)$ (usually $N=N_+(t)$ for a node $t\in V(T^A)$). For all $i\in[m]$, let
$P_i$ be the set of all $u\in N$ such that $\tp_q^\MSO(A_u,\le)=\theta_i$. We
call $(P_1,\ldots,P_m)$ the \emph{type partition} of $N$. (Note that some of the
$P_i$ may be empty. We always allow partitions to have empty parts.) The
following lemma extends classical type-composition theorems
\cite{Makowsky2004,FefermanV1959} to our situation, where substructures are
combined through b-nodes.

\begin{lemma}[Ordered type composition at b-nodes]
  \label{lem:btype}
  For every $\theta\in\Theta$ there is an
  $\MSO[\tau^\dbstar]$-formula
    $\btype_\theta(X_1,\ldots,X_m)$
  such that for every otxx $A$, every b-node $t\in V(T^A)$, and every compatible
  linear order $\le$ of $A$, if $(P_1,\ldots,P_m)$ is the type partition of
  $N_+(t)$, then
  \begin{align*}
    A_{(t)}\models\btype_{\theta}(P_1,\ldots,P_m)
    \text{ if, and only if, }
    \tp_q^{\MSO}(A_t,\le)=\theta.    
  \end{align*}
\end{lemma}
\begin{proof}
  For $0\le i\le q$, let $\Theta_i:=\TP^{\MSO}(\tau^\dbstar\cup\{\le\} ,q-i,i)$,
  and suppose that 
  $\Theta_i=\{\theta_{i1},\ldots,\theta_{im_i}\}$. Then
  $\Theta_0=\Theta$ and $m_0=m$, and we may assume that
  $\theta_{0j}=\theta_j$ for all $j\in[m]$. Let $q' := 1+\sum_{i=1}^q(1+m_i)$. 
  The core of the proof is the following claim.

  \begin{claim*}
    Let $A,B$ be otxxs and $\le^A,\le^B$ compatible
    linear orders of $A,B$, respectively. Let
    $t\in V(T^A)$ and $t'\in V(T^B)$. Let $(P_{01},\ldots,P_{0m_0})$ and $(Q_{01},\ldots,Q_{0m_0})$ be the
    type partitions of $N_+(t)$ and $N_+(t')$, respectively.
    If
    \begin{equation}
      \label{eq:6}
      \tp_{q'}^{\MSO}(A_{(t)},P_{01},\ldots,P_{0m_0})=\tp_{q'}^{\MSO}(B_{(t')},Q_{01},\ldots,Q_{0m_0}),
    \end{equation}
    then $(A_t,\le^A)\equiv_q^\MSO(B_{t'},\le^B)$.      
    \proof
    We shall prove that Duplicator has a winning strategy for the $q$-move
    $\MSO$-game on $(A_t,\le^A$), $(B_{t'},\le^B)$. It is crucial to note that the
    compatible linear orders $\le^A,\le^B$ coincide with the partial orders
    $\preceq^A,\preceq^B$ of the structures $A,B$ when restricted to
    $U(A_{(t)}),U(B_{(t)})$, respectively. The reason for this is that the
    restrictions of $\preceq^A,\preceq^B$ to $U(A_{(t)}),U(B_{(t)})$, respectively,
    are linear orders, because $t$ and $t'$ are b-nodes. This means that the games
    on $(A_{(t)},\le^A),(B_{(t')},\le^B)$ and on $A_{(t)},B_{(t')}$ are the same.

    With every sequence $\bar P=(P_1,\ldots,P_p)$ of subsets of
    $U(A_t)$ we associate a sequence 
    \begin{align*}
       \bar P^+:=(P_{01},\ldots,P_{0m_0},P_{10},P_{11},\ldots,P_{1m_1},\allowbreak P_{20},\;\ldots\;,P_{(p-1)m_{p-1}}P_{p0},P_{p1},\ldots,P_{pm_p})  
    \end{align*}
    of subsets of $U(A_{(t)})$ as follows:
    \begin{itemize}
    \item $P_{i0}:=P_i\cap U(A_{(t)})$, for all
      $i\in[p]$;
    \item $P_{ij}$ is the set of $u\in N_+(t)$ with
      $\theta_{ij} = \tp^\MSO_{q-i}\big(A_u,\le,P_1\cap U(A_u),\ldots,P_i\cap
      U(A_u)\big) $ for all $i\in[p]$, $j\in[m_i]$. 
    \end{itemize}
    For every sequence $\bar Q=(Q_1,\ldots,Q_p)$ of subsets of
    $U(B_{t'})$ we define $\bar Q^+$ similarly, and for every
    position $\Pi=(P_i,Q_i)_{i\in[p]}$ of the \MSO-game on
    $(A_t,\le^A),(B_{t'},\le^B)$ we let $\Pi^+$ be the
    position of the \MSO-game on $A_{(t)},B_{(t)}$ consisting of $\bar P^+$ and $\bar Q^+$. 

    Our goal is to define a strategy for Duplicator in the $q$-move
    game on $(A_t,\le^A),(B_{t'},\le^B)$ such that for
    every reachable position $\Pi$ of length $p$ the position $\Pi^+$
    is a $1+\sum_{i=p+1}^q(1+m_i)$-move winning position for Duplicator
    in the \MSO-game on $A_{(t)},B_{(t')}$. Such a
    strategy will clearly be a winning strategy. We define the
    strategy inductively. For the initial empty position $\Pi_0$ we have
    $\Pi_0^+=(P_{0j},Q_{0j})_{j\in[m_0]}$, and it follows from \eqref{eq:6} that
    is is a $q'$-move winning position for Duplicator in the \MSO-game on
    $A_{(t)},B_{(t')}$. 

    So suppose now we are in a position $\Pi=(P_i,Q_i)_{i\in[p]}$ and
    the corresponding position $\Pi^+$ is a
    $1+\sum_{i=p+1}^q(1+m_i)$-move winning position for Duplicator in
    the \MSO-game on $A_{(t)},B_{(t')}$. Without loss
    of generality, we assume that
    in the $(p+1)$st move of the game on
    $(A_t,\le^A),(B_{t'},\le^B)$, Spoiler chooses a
    set $P_{p+1}\subseteq U(A_t)$. (The case that he chooses a
    set $Q_{p+1}\subseteq U(B_{t'})$ is symmetric.)

    We define the sets $P_{ij}$ for $i\in[p+1]$ and
    $j\in\{0,\ldots,m_i\}$ as above. Suppose that, starting in
    position $\Pi^+$, in the game on
    $A_{(t)},B_{(t')}$ Spoiler selects the sets
    $P_{(p+1)0},\ldots,P_{(p+1)m_{p+1}}$ in the next $m_{p+1}+1$
    moves. Let $Q_{(p+1)0},\ldots,Q_{(p+1)m_{p+1}}$ be Duplicator's
    answers according to some winning strategy. Let $(\Pi^+)'$ be the
    resulting position of the \MSO-game on
    $A_{(t)},B_{(t')}$; this is a
    $1+\sum_{i=p+2}^q(1+m_i)$-move winning position for Duplicator.
    
    As the sets $P_{(p+1)0},\ldots,P_{(p+1)m_{p+1}}$ form a partition
    of $N_+(t)$, the sets $Q_{(p+1)1},\ldots,Q_{(p+1)m_{p+1}}$ form a
    partition of $N_+(t')$, because otherwise Spoiler wins in the
    next round of the game (this explains the '$1+$' in the the number
    of moves of the game). Let $u'\in N_+(t')$ and $j=j(u')$ such that
    $u'\in Q_{(p+1)j}$. Then there is at least one $u\in
    P_{(p+1)j}$; otherwise Spoiler wins in the next round of the game.
    Let $j'\in[m_p]$ such that $u\in P_{pj'}$. Then
    \begin{align}
      \label{eq:20}
      \tp_{q-p}(A_{u},\le,P_1\cap
      U(A_u),\ldots,P_p\cap U(A_u))&=\theta_{pj'},\\
      \label{eq:21}
      \tp_{q-p-1}(A_{u},\le,P_1\cap U(A_u),\ldots,P_{p+1}\cap U(A_u))&=\theta_{(p+1)j}.
    \end{align}
    Hence the type $\theta_{pj'}$ is the unique ``restriction'' of
    $\theta_{(p+1)j}$, and for all $u''\in P_{(p+1)j}$ we have
    $u''\in P_{pj'}$. This implies that $u'\in Q_{pj'}$, because
    otherwise Spoiler wins in the next round of the game. It follows
    that
    \begin{equation}
      \label{eq:22}
      \tp_{q-p}(B_{u'},\le,Q_1\cap U(B_{u'}),\ldots,Q_p\cap U(B_{u'}))=\theta_{pj'}.
    \end{equation}
    This implies that there is a $Q_{(p+1)}^{u'}\subseteq
    U(B_{u'})$ with
    \begin{align*}
      \theta_{(p+1)j} & = \tp_{q-p-1}(B_{u'},\le,Q_1\cap
    U(B_{u'}),\ldots,Q_p\cap
    U(B_{u'}),Q_{(p+1)}^{u'})\, .
    \end{align*}
    We let
    $
    Q_{p+1}:=Q_{(p+1)0}\cup\bigcup_{u'\in N_+(t')}Q_{(p+1)}^{u'}.
    $
    The new position is $\Pi':=(P_i,Q_i)_{i\in[p+1]}$. Then
    $ (\Pi')^+=(\Pi^+)'$, which is a $1+\sum_{i=p+2}^q(1+m_i)$-move
    winning position for Duplicator in the \MSO-game on
    $A_{(t)},B_{(t')}$.
    \uend
  \end{claim*}

  The claim implies that $\tp_q^\MSO(A,\le^A)$ only depends on the type 
  of $\tp_{q'}^\MSO(A_{(t)},P_1,\ldots,P_m)$. Let
  $\theta\in\Theta$. To define the formula $\btype_\theta$, let
  $\theta_1',\ldots,\theta_\ell'$ be the list of all types $\theta'\in\TP^\MSO(\tau,q',m)$  such that
  $\tp_{q'}^\MSO(A_{(t)},P_1,\ldots,P_m)=\theta'$ implies $\tp_q^\MSO(A,\le^A)=\theta$.
  Then $\tp_q^\MSO(A,\le^A)=\theta$ if, and only, if
  \begin{align*}
    A_{(t)}\models\bigvee_{i=1}^\ell\bigwedge_{\psi(X_1,\ldots,X_m)\in\theta'_i}\psi(P_1,\ldots,P_m)\,
    . 
  \end{align*}
\end{proof}

Note that the vocabulary of the formula $\btype$ in the lemma is $\tau^\dbstar$
and not $\tau^\dbstar\cup\{\le\}$. It will be important throughout the proofs of
the lifting theorems to keep track of the vocabularies.  The next lemma is a
similar result for a-nodes, but there is one big difference: the formula
$\atype$ we obtain has vocabulary $\{\le\}$ and not $\tau^\dbstar$. This means
that, at least a priori, the formula is not order-invariant. For b-nodes, the
formula $\btype_{\theta}$ does not depend on the order, because for b-nodes $t$
every compatible linear order $\le$ coincides with $\preceq$ on
$U(A_{(t)})$. The proof of the lemma is a straightforward adaptation of the
proof of the previous lemma.

\begin{lemma}[Ordered type composition at a-nodes]
  \label{lem:atype}
  For every $\theta\in\Theta$ there is an $\MSO[\{\le\}]$-formula
  $\atype_\theta(X_1,\ldots,X_m)$ such that for every otxx 
  $A$, every a-node $t\in V(T^A)$, and every compatible linear order $\le$
  of $A$, if $(P_1,\ldots,P_m)$ is the type partition of $N_+(t)$,
  then 
  \begin{align*}
    (N_+(t),\le)\models\atype_{\theta}(P_1,\ldots,P_m)
    \text{ if, and only if, }
    \tp_q^{\MSO}(A_t,\le)=\theta.    
  \end{align*}
\end{lemma}

\subsection{Order-Invariant Type Compositions}
\label{sec:inv-composition}

Recall from Section~\ref{sec:games-types} the definition of order-invariant
types and the characterization of order-invariant equivalence that we gave in
Lemma~\ref{lem:oitypes}. We continue to adhere to the assumptions made in the
previous subsections (otxx have segmented tree decompositions of adhesion at
most $k$, $q$ is sufficiently large, and 
$\TP^\MSO(\tau^\dbstar \cup \{\leq\},q) = \Theta=\{\theta_1,\ldots,\theta_m\}$) and use the
same notation.

Recall that, since $q$ is sufficiently large and the class of otxxs is
\MSO-definable, if $A$ is an otxx and $A'\equiv_q^\MSO A$ then $A'$ is an
otxx. This implies that if $A\equiv_q^\OIMSO A'$, then all structures appearing
in a sequence witnessing this equivalence
(cf.~Lemma~\ref{lem:oitypes}\eqref{item:witness}) are otxxs. The same is true
for sub-otxxs. However, it is not clear that all linear orders appearing in such
a witnessing sequence are compatible. In other words, it is not clear that order
invariance on otxxs coincides with invariance with respect to all compatible
orders. For this reason, we need to introduce a finer equivalence relation
$\equiv_{co}$, \emph{compatible-order equivalence}. For two sub-otxx $A,A'$, we
let $A\equiv_{co} A'$ if there is a sequence $A_0,\ldots,A_\ell$ of sub-otxxs
and compatible linear orders $\le_i,\le_i'$ of $A_i$ such that $A=A_0$ and
$A'=A_\ell$ and $(A_{i-1},\le_{i-1})\equiv_q^{\MSO}(A_i,\le_i')$ for all
$i\in[\ell]$. Then clearly $A\equiv_{co} A'$ implies $A\equiv_q^\OIMSO A'$. The
converse holds as well, because from an arbitrary linear order we can define a
compatible linear order, but this is not important for us.

Let us call a type $\theta\in\Theta$ \emph{realizable} if there is a sub-otxx
$A$ and a compatible linear order $\le$ of $A$ with
$\tp_q^\MSO(A,\le)=\theta$. We call $(A,\le)$ a \emph{realization} of $\theta$.
Two types $\theta,\theta'\in\Theta$ are \emph{compatible-order equivalent} (we
write $\theta\equiv_{co}\theta'$) if there are realizations $(A,\le)$ of
$\theta$ and $(A',\le')$ of $\theta'$ such that $A\equiv_{co}A'$. Then
$\equiv_{co}$ is an equivalence relation on the set of realizable types. We
denote the equivalence class of a type $\theta\in\Theta$ by
$\angles{\theta}_{co}$. Clearly, we have
$\angles{\theta}_{co}\subseteq\angles{\theta}$.

Now let $A$ be an otxx and $t\in V(T^A)$. We call a set $\Theta'\subseteq\Theta$
\emph{compatible at $t$} if there is a compatible linear order $\le$ of $U(A_t)$
such that $\theta:=\tp_q^\MSO(A_t,\le)\in\Theta'$ and
$\Theta'\subseteq\angles{\theta}_{co}$. Note that this implies that all
$\theta'\in\Theta'$ are realizable.

A \emph{cover} of a set $N$ is a sequence $(P_1,\ldots,P_m)$ of subsets of $N$
such that $\bigcup_{i=1}^m P_i=N$. For an otxx~$A$ and node $t\in V(T^A)$, we
call a cover $(P_1,\ldots,P_m)$ of $N_+(t)$ \emph{compatible} if for all $u\in
N_+(t)$ the set $\{\theta_i\mid i\in[m]\text{ such that }u\in P_i\}$ is
compatible at $u$. Observe that if $(P_1,\ldots,P_m)$ is the type partition of
$N_+(t)$ with respect to some compatible linear order, then $(P_1,\ldots,P_m)$
is a compatible cover.

\begin{lemma}[Order-invariant type composition at b-nodes]
  \label{lem:oibtype}
  For every $\theta\in\Theta$ there is an
  $\MSO[\tau^\dbstar]$-formula $\oibtype_\theta(X_1,\ldots,X_m)$ such that for
  every otxx $A$, every b-node $t\in V(T^A)$, and every compatible cover
  $(P_1,\ldots,P_m)$ of $N_+(t)$, the set of all $\theta\in\Theta$ with
  $A_{(t)}\models\oibtype_{\theta}(P_1,\ldots,P_m)$ is
  compatible at $t$.
\end{lemma}

The idea of the proof is that within the
structure $A_{(t)}$ we can quantify over the possible type partitions
of the children (they are just collections of sets) and then apply
Lemma~\ref{lem:btype} to each of them individually.

\begin{proof}[Proof of Lemma~\ref{lem:oibtype}]
  Let $\phi(X_1,\ldots,X_m,Y_1,\ldots,Y_m)$ be an \MSO-\-formula
  stating that $Y_i\subseteq X_i$ for all $i$, that the $Y_i$ are
  mutually disjoint, and that $\bigcup_i Y_i=\bigcup_i X_i$. 
  We let 
  \begin{align*}
    \oibtype_\theta(X_1,\ldots,X_m):= \exists Y_1\ldots\exists
    Y_m\big(\phi(X_1,\ldots,X_m,Y_1,\ldots,Y_m)\wedge\btype_\theta(Y_1,\ldots,Y_m)\big)\, . 
  \end{align*}
  Let $A$ be an otxx, $t\in V(T^A)$ a b-node, and 
  $(P_1,\ldots,P_m)$ a compatible cover of $N_+(t)$.
  Let $\Theta^t$ be the set of all $\theta$ such that
  $A_{(t)}\models\oibtype_{\theta}(P_1,\ldots,P_m)$. We need to prove
  that $\Theta^t$ is compatible at $t$.
  
  For every $u\in N_+(t)$, let $\Theta^u:=\{\theta_i\mid
  i\in[m]\text{ such that }u\in P_i\}$. As the cover
  $(P_1,\ldots,P_m)$ is compatible, for all $u$ the set $\Theta^u$ is
  compatible at $u$. Thus there is a
  $\theta_u\in\Theta^u$ and a compatible linear order $\le_u$ of
  $A_u$ such that
  $\theta_u=\tp_q^\MSO(A_u,\le_u)$ and $\Theta^u\subseteq\angles{\theta_u}_{co}$. Let $\le$ be the (unique)
  compatible linear order of $A_t$ such that for all $u\in
  N_+(t)$, the restriction of $\le$ to
  $U(A_u)$ is $\le_u$. For every $i\in[m]$, let $Q_i$ be the
  set of all $u\in N_+(t)$ such that $\theta_u=\theta_i$.  Then
  $(Q_1,\ldots,Q_m)$ is a partition of $N_+(t)$ that refines the
  cover $(P_1,\ldots,P_m)$.

  Let $\theta_t:=\tp_q^\MSO(A_t,\le)$. By Lemma~\ref{lem:btype}, we have
  $A_{(t)}\models\btype_{\theta^t}(Q_1,\ldots,Q_m)$ and, thus,
  $A_{(t)}\models\oibtype_{\theta^t}(Q_1,\ldots,Q_m)$. Hence
  $\theta_t\in\Theta^t$.

  We claim that $\Theta^t\subseteq\angles{\theta_t}_{co}$. Let
  $\theta\in\Theta^t$. We first prove that $\theta$ is realizable. Since we have 
  $A_{(t)}\models\oibtype_\theta(P_1,\ldots,P_m)$, there is a partition
  $(Q'_1,\ldots,Q'_m)$ of $N_+(t)$ that refines the cover $(P_1,\ldots,P_m)$ such
  that 
  \begin{equation}\label{eq:oib1}
    A_{(t)}\models\btype_{\theta}(Q'_1,\ldots,Q'_m).
  \end{equation}
  For each $u\in N_+(t)$,
  let $\theta_u':=\theta_i$ for the unique $i$ such that $u\in Q'_i$. Then
  $\theta_u'\in\Theta^u$, and thus $\theta_u'$ is realizable.  Let $(A_u',\le_u')$
  be a realization of $\theta_u'$.

  Let $A'$ be the sub-otxx obtained from $A_t$ by simultaneously replacing the
  sub-otxx $A_u$ by the sub-otxx $A_u'$ for all $u\in N_+(t)$ (see
  page~\pageref{page:replacing} for a description of the replacement
  operation). As $\theta_u'\in\Theta^u\subseteq\angles{\theta_u}_{co}
  \subseteq\angles{\theta_u}$, 
  we have $A_u\equiv_q^\MSO A_u'$ and thus the induced substructures
  $A[\{u\}\cup\sigma^A(u)]$ and $A_u'[\{u'\}\cup\sigma^{A_u'}(t')]$,
  where $u'$ is the root of $A_u'$, are isomorphic, and the
  replacement is possible. (We will use similar arguments about
  replacements below without mentioning them explicitly.)  Let $\le'$
  be the (unique) compatible linear order of $A'$ such that for all
  $u\in N_+(t)$, the restriction of $\le'$ to $U(A_u')$ is
  $\le_u'$. Note that $(A'_{(t)},\le')=(A_{(t)},\le)$, because the
  linear orders $\le$ and $\le'$ both coincide with $\preceq^A$ on
  $U(A_{(t)})$. Thus by \eqref{eq:oib1},
  $A'_{(t)}\models\btype_{\theta}(Q'_1,\ldots,Q'_m)$, and by
  Lemma~\ref{lem:btype}, $\tp_q^\MSO(A',\le')=\theta$. Thus $\theta$
  is realizable.

  It remains to prove that $\theta_t\equiv_{co}\theta$. For each $u\in N_+(t)$, we
  have
  $\tp_q^\MSO(A_u,\le_u)=\theta_u\equiv_{co}\theta_u'=\tp_q^\MSO(A'_u,\le'_u)$. Thus
  there is a sequence $A_{u0},\ldots,A_{u\ell}$ of sub-otxxs and for
  each $i$ two compatible
  linear orders $\le_{ui},\le_{ui}'$ of $A_{ui}$ such that
  $(A_{u0},\le_{u0})=(A_u,\le_u)$ and
  $(A_{u\ell},\le_{u\ell})=(A_u',\le_u')$ and 
  \begin{align*}
    \tp_q^\MSO(A_{u(i-1)},\le_{u(i-1)}')=\tp_q^\MSO(A_{ui},\le_{ui})    
  \end{align*}
  for all $i\in[\ell]$. As we do not require
  the $A_{ui}$ and the orders $\le_{ui},\le_{ui}'$ to be
  distinct, we may assume without loss of generality that the
  sequences have the same length $\ell$ for all $u$. Let $A_i$ be the structure obtained
  from $A_t$ by simultaneously replacing $A_u$ by $A_{ui}$
  for all $u\in N_+(t)$. Define linear orders $\le_i,\le_i'$ of
  $A_i$ from the orders
  $\le_{ui}',\le_{ui}$ and $\preceq^A$ in the usual way. The resulting sequence of
  structures and orders witnesses 
  $\theta_t=\tp_q^\MSO(A_t,\le)\equiv_{co}\tp_q^\MSO(A',\le')=\theta$. To
  prove this, we apply Lemma~\ref{lem:btype} at every step.
\end{proof}




\begin{lemma}[Order-invariant type composition at a-nodes]
  \label{lem:oiatype}
  For every $\theta\in\Theta$ there is an $\CMSO[\emptyset]$-formula
  $\oiatype_\theta(X_1,\ldots,X_m)$ such that for every otxx $A$, every a-node
  $t\in V(T^A)$, and every compatible cover $(P_1,\ldots,P_m)$ of $N_+(t)$, the
  set of all $\theta\in\Theta$ with
  $(N_+(t))\models\oiatype_{\theta}(P_1,\ldots,P_m)$ is compatible at $t$.
\end{lemma}

Here $(N_+(t))$ denotes the $\emptyset$-structure with universe $N_+(t)$.  Note
that, as opposed to the formula $\atype_\theta$ of Lemma~\ref{lem:atype}, the
formula $\oiatype_\theta$ has an empty vocabulary. Thus, the condition expressed
by this formula no longer depends on the arbitrarily chosen compatible linear
order.  The proof builds on the ideas developed in the previous proofs and, in
addition, crucially depends on the fact that \OIMSO\ coincides with \CMSO\ on
set structures, which only have monadic relations.

\begin{proof}[Proof of Lemma~\ref{lem:oiatype}]
  Let $\theta\in\Theta$. 
  We may view the $\MSO$-formula $\atype_\theta(X_1,\ldots,X_m)$ as an
  $\MSO$-sentence of vocabulary 
  $
  \sigma:=\{\le,X_1,\ldots,X_m\},
  $
  where we interpret the $X_i$ as unary relation symbols. Let
  $\chi_\theta^1$ be the conjunction of this sentence with a sentence
  saying that $\le$ is a linear order and the $X_i$ partition the
  universe. Then all models of $\chi_\theta^1$ are proper word
  structures. Let $q_1$ be an upper bound for the quantifier rank of
  the formulas $\chi_{\theta'}^1$ for $\theta'\in\Theta$. Let
  $\Xi:=\TP^\MSO(\sigma,q_1)$, and for each $\xi\in\Xi$, let
  $\angles{\xi}$ be the order-invariant type that contains $\xi$. Now
  let $\xi_1,\ldots,\xi_\ell$ be all $\xi\in\Xi$ that contain
  $\chi_\theta^1$, and let
  \begin{align*}
    \chi_\theta^2:=\bigvee_{i=1}^\ell\bigvee_{\xi\in\angles{\xi_i}}\bigwedge_{\phi\in\xi}\phi\, .    
  \end{align*}
  Then $\chi_\theta^2$ is order-invariant; we may view it has the ``best
  order-invariant approximation'' of $\chi_\theta^1$. The sentence $\chi_\theta^2$
  is over the vocabulary of words, but is invariant with respect to the ordering
  underlying the word. In other words, it is an order-invariant formula of
  vocabulary $\{X_1,\ldots,X_m\}$ and, thus, equivalent to a $\CMSO$-sentence
  $\chi_\theta^3$ over the same vocabulary~\cite[Corollary
  4.3]{Courcelle1996}.

  We view $\chi_\theta^3=\chi_{\theta}^3(X_1,\ldots,X_m)$ as a $\CMSO$-formula
  of empty vocabulary with free variables $X_1,\ldots,X_m$.

  Let $\Theta_\theta$ be the set of all
  $\theta'\in\Theta$ such that the following holds: there is an otxx
  $A'$, an a-node $t'\in V(T^{A'})$, and a compatible linear order
  $\le'$ of $A'$ such that $(N_+(t'))\models
  \chi_{\theta}^3(P_1',\ldots,P_m')$ for the type partition
  $(P_1',\ldots,P_m')$ of $N_+(t')$
  and
  $\tp_q^\MSO(A'_{t'},\le')=\theta'$.
  Then trivially, all $\theta'\in\Theta_{\theta}$ are
  realizable. 

  \begin{claim}\label{cl:oia2}
    If\/ $\Theta_{\theta}\neq\emptyset$, then $\theta$ is realizable and
    $\theta\in\Theta_{\theta}$ and
    $\Theta_{\theta}\subseteq\angles{\theta}_{co}$.

    \proof
    Let $\theta'\in\Theta_{\theta}$. Let $A'$ be an otxx, $t'\in
    V(T^{A'})$ an a-node, 
    $\le'$ a compatible linear order of $A'$, and $(P_1',\ldots,P_m')$
    the type partition of $N_+(t')$ such that 
    $(N_+(t'))\models \chi_{\theta}^3(P_1',\ldots,P_m')$ and
    $\tp_q^\MSO(A'_{t'},\le')=\theta'$. Then
    $(N_+(t'),\le')\models \chi_{\theta}^2(P_1',\ldots,P_m')$. 
    Hence there is a $(N,\le)$ and a partition $P_1,\ldots,P_m$ of $N$
    such that 
    \begin{align*}
      (N,\le,P_1,\ldots,P_m)\equiv_{q_1}^\OIMSO(N_+(t'),\le',P_1',\ldots,P_m')      
    \end{align*}
    and $(N,\le,P_1,\ldots,P_m)\models\chi_\theta^1$. Equivalently, we have 
    $(N,\le)\models\atype_{\theta}(P_1,\ldots,P_m)$.

    By Lemma~\ref{lem:oitypes}, there is an $\ell\in\NN$ and for $0\le
    i\le\ell$ sets
    $N_i$, partitions $(P_{i1},\ldots,P_{im})$ of $N_i$, and linear
    orders $\le_i,\le_i'$ of $N_i$ such that
    $(N_0,\le_0,P_{01},\ldots,P_{0m})=(N,\le,P_1,\ldots,P_m)$ and
    $(N_\ell,\le_\ell',P_{\ell 1},\ldots,P_{\ell
      m})=(N_+(t'),\le',P_1',\ldots,P_m')$ and
    $(N_{i-1},\le_{i-1}',P_{(i-1)1},\ldots,P_{(i-1)m})\equiv_{q_1}^\MSO(N_i,\le_i,P_{i1},\ldots,P_{im})$. 

    We let $A^\ell:=A'_{t'}$ and $t_\ell:=t'$, and for $0\le i<\ell$
    we build a sub-otxx $A^i$ as follows: we take a fresh node $t_i$,
    which will be the root of the tree $T^{A^i}$. We make
    $N_+(t_i):=N_i$ the set of children of $t_i$. The node $t_i$ will
    be an a-node in $A^i$. We let
    $\beta^{A^i}(t_i):=\beta^{A'}(t')$. For each $u\in N_i$, say, with
    $u\in P_{ij}$, we take some $u'\in P_j'$. Note that $P_j'$ is
    nonempty, because $P_{ij}$ is nonempty and 
    $(N_i,P_{i1},\ldots,P_{im})\equiv_{q_1}^\MSO(N_+(t'),P_1',\ldots,P_m')$.
    Then we take a copy
    $A^i_{u}$ of $A'_{u'}$ and identify the
    copy of $u'$ with $u$ and the copy of $\sigma^{A'}(u')$ with the
    corresponding elements in $\beta^{A^i}(t_i)=\beta^{A'}(t')$.
    We define two compatible orders $\le_i,\le_i'$ on $A^i$ that
    extend the corresponding orders on $N_i$ and coincide with the
    linear order induced by $\le'$ on the copies of the sub-otxxs
    $A_{u'}'$ that we used to build $A^i$.

    Then for $0\le i<\ell$, all $j\in[m]$, and all $u\in N_i$, if
    $u\in P_{ij}$ then $(A^i_u,\le_i)$ and $(A^i_u,\le_i')$ are copies
    of $(A'_{u'},\le')$ for some $u'\in P_j'$, and hence 
    $\tp_q^\MSO(A^i_u,\le_i)=\tp_q^\MSO(A^i_u,\le_i')=\tp_q^\MSO(A'_{u'},\le')=\theta_j$.      
    Since $(N_{i-1},\le_{i-1}',P_{(i-1)1},\ldots,P_{(i-1)m})\equiv_{q_1}^\MSO (N_i,\le_i,P_{i1},\ldots,P_{im})$,        
    it follows from Lemma~\ref{lem:atype} that we have
    $\tp_q^\MSO(A^{i-1},\le_{i-1}')=\tp_q^\MSO(A^i,\le_i)$ for all
    $i$. Moreover, as we have 
    $(N_0,\le_0)\models\atype_{\theta}(P_{01},\ldots,P_{0m})$, again
    by Lemma~\ref{lem:atype} we have
    $\tp_q^\MSO(A^0,\le_{0})=\theta$. 

    This implies that $\theta$ is realizable and that
    $\theta\equiv_{co}\theta'$, or equivalently,
    $\theta'\in\angles{\theta}_{co}$. As this holds for all
    $\theta'\in\Theta_{\theta}$, we have
    $\Theta_{\theta}\subseteq\angles{\theta}_{co}$.  We have
    $\theta\in\Theta_{\theta}$ because
    $(N_0,\le_0)\models\atype_{\theta}(P_{01},\ldots,P_{0m})$ implies
    $(N_0,\le_0,P_{01},\ldots,P_{0m})\models\chi^2_\theta$, and this
    implies $(N_0)\models \chi^3_\theta(P_{01},\ldots,P_{0m})$.  \uend
  \end{claim}

  \begin{claim}\label{cl:oia3}
    Let $A$ be an otxx, $t\in V(T^A)$ an a-node, $\le$ a compatible linear order
    of $A$, and $(P_1,\ldots,P_m)$ the type partition of $N_+(t)$. Then the set
    $\Theta_t$ of all $\theta\in\Theta$ with
    $(N_+(t))\models\chi_\theta^3(P_1,\ldots,P_m)$ is compatible at $t$.

    \proof
    Let $\theta_t:=\tp_q^\MSO(A_t,\le)$. Then for all 
    $\theta\in\Theta_t$ we have $\theta_t\in\Theta_{\theta}$ and thus,
    by Claim~\ref{cl:oia2}, 
    $\theta_t\in\angles{\theta}_{co}$. As
    $\equiv_{co}$ is an equivalence relation, it follows that
    $\angles{\theta_t}_{co}=\angles{\theta}_{co}$. Thus
    $\Theta_t\subseteq \angles{\theta_t}_{co}$, and this shows
    that $\Theta_t$ is compatible at $t$.
    \uend
  \end{claim}

  The rest of the proof is very similar to the proof of
  Lemma~\ref{lem:oibtype}. Again, we let $\phi(X_1,\ldots,X_m,Y_1,\ldots,Y_m)$
  be an \MSO-formula
  stating that $Y_i\subseteq X_i$ for all $i$, that the $Y_i$ are
  mutually disjoint, and that $\bigcup_i Y_i=\bigcup_i X_i$. We let 
  $\oiatype_\theta(X_1,\ldots,X_m):=\exists Y_1\ldots\exists
  Y_m\big(\phi(X_1,\ldots,X_m,Y_1,\ldots,Y_m)\wedge\chi^3_\theta(Y_1,\ldots,Y_m)\big)$.

  Let $A$ be an otxx, $t\in V(T^A)$ an a-node, and $(P_1,\ldots,P_m)$
  a compatible cover of $N_+(t)$. Let $\Theta^t$ be the set of all
  $\theta\in\Theta$ such that $(N_+(t))\models
  \oiatype_\theta(P_1,\ldots,P_m)$. We need to prove that $\Theta^t$
  is compatible at $t$.

  For every $u\in N_+(t)$, let $\Theta^u:=\{\theta_i\mid
  i\in[m]\text{ such that }u\in P_i\}$. As the cover
  $(P_1,\ldots,P_m)$ is compatible, for all $u$ the set $\Theta^u$ is
  compatible at $u$. In particular, there is a
  $\theta_u\in\Theta^u$ and a compatible linear order $\le_u$ of
  $A_u$ such that
  $\theta_u=\tp_q^\MSO(A_u,\le_u)$ and $\Theta^u\subseteq\angles{\theta_u}_{co}$. Let $\le^1$ be a
  compatible linear order of $A_t$ such that for all $u\in
  N_+(t)$, the restriction of $\le^1$ to
  $U(A_u)$ is $\le_u$. For every $i\in[m]$, let $Q_i$ be the
  set of all $u\in N_+(t)$ such that $\theta_u=\theta_i$.  Then
  $(Q_1,\ldots,Q_m)$ is the type partition of $N_+(t)$ in $(A_t,\le^1)$,
  and it refines the
  cover $(P_1,\ldots,P_m)$.

  By Claim~\ref{cl:oia3}, the set $\Theta_t(Q_1,\ldots,Q_m)$ of all
  $\theta\in\Theta$ such that
  $(N_+(t))\models\chi_\theta^3(Q_1,\ldots,Q_m)$ is compatible at
  $t$. Thus there is a type $\theta_t\in \Theta_t(Q_1,\ldots,Q_m)$ and
  a linear order $\le^2$ of $A$ such that
  $\tp_q^\MSO(A_t,\le^2)=\theta_t$ and
  $\Theta_t(Q_1,\ldots,Q_m)\subseteq\angles{\theta_t}_{co}$. As
  $\theta_t\in \Theta_t(Q_1,\ldots,Q_m)$ we have
  $(N_+(t))\models\chi_{\theta_t}^3(Q_1,\ldots,Q_m)$ and thus
  $(N_+(t))\models\oiatype_{\theta_t}(P_1,\ldots,P_m)$. Thus
  $\theta_t\in\Theta^t$. 

  We need to prove that $\Theta^t\subseteq\angles{\theta_t}_{co}$.
  Let
  $\theta\in\Theta^t$. Then
  $A_{(t)}\models\oiatype_\theta(P_1,\ldots,P_m)$, and thus there is a
  partition $(Q'_1,\ldots,Q'_m)$ of $N_+(t)$ that refines the cover
  $(P_1,\ldots,P_m)$ such that
  $(N_+(t))\models\chi^3_{\theta}(Q'_1,\ldots,Q'_m)$. Let
  $\Theta_t(Q_1',\ldots,Q_m')$ be the set of all 
  $\theta'\in\Theta$ such that
  $(N_+(t))\models\chi_{\theta'}^3(Q_1',\ldots,Q_m')$. Then we have
  $\theta\in\Theta_t(Q_1',\ldots,Q_m')$. By Claim~\ref{cl:oia3}, the set
  $\Theta_t(Q_1',\ldots,Q_m')$ is compatible at $t$. Thus there is a
  $\theta_t'\in \Theta_t(Q_1',\ldots,Q_m')$ and a compatible linear order $\le^3$
  of $A$ such that $\tp_q^\MSO(A_t,\le^3)=\theta_t'$ and
  $\Theta_t(Q_1',\ldots,Q_m')\subseteq\angles{\theta_t'}_{co}$. 

  It
  remains to prove that 
  $
  \theta_t\equiv_{co}\theta_t',
  $
  because then
  $\theta\in\Theta_t(Q_1',\ldots,Q_m')\subseteq\angles{\theta_t'}_{co}=\angles{\theta_t}_{co}$. 
  For each $u\in
  N_+(t)$, let $\theta_u:=\tp_q^\MSO(A_u,\le^2)$ and
  $\theta_u':=\tp_q^\MSO(A_u,\le^3)$. Then $\theta_u=\theta_i$ for the
  unique $i$ such that $u\in Q_i$ and $\theta_u'=\theta_{i'}$ for the
  unique $i'$ such that $u\in Q_{i'}$. 
  As both $(Q_1,\ldots,Q_m)$ and $(Q_1',\ldots,Q_m')$ refine the cover
  $(P_1,\ldots,P_m)$ and the set $\Theta^u$
  is compatible at $u$, we have $\theta_u\equiv_{co}\theta_u'$. Now we
  can form a sequence witnessing $\theta_t\equiv_{co}\theta_t'$ from
  sequences witnessing $\theta_u\equiv_{co}\theta_u'$ for the $u\in
  N_+(t)$  as in the proof of Lemma~\ref{lem:oibtype}
  (when we showed $\theta_t\equiv_{co}\theta$).
\end{proof}

\subsection{Proofs of the lifting theorems}
\label{sec:proofs-lifting}

\begin{proof}[Proof of Theorem~\ref{th:lift-oimso}]
  Let $\CC$ be a class of structures over some vocabulary $\tau$ that admit
  \CMSO-definable ordered tree decompositions and let $\phi$ be an
  \OIMSO-formula over $\tau$. We show that there exists a \CMSO-formula $\psi$,
  such that for every structure $A$ from $\CC$ we have $A \models \phi$ if, and
  only if, $A \models \psi$.

  First of all, we turn $A$ into a structure $A^\dbstar$ that is isomorphic to
  an otx of $A$. Using the theorem's precondition, this is possible by a
  \CMSO-transduction that produces otxs with bounded adhesion. Using the
  transformations discussed in Section~\ref{sec:segmented}, we continue to turn
  $A^\star$ into an otx whose tree decomposition is segmented and, then, expand it
  into an otxx $A^\dbstar$. Both transductions preserve the bounded adhesion
  property. Since $A$'s relations are still present in $A^\dbstar$ and we can
  distinguish the elements in $A^\dbstar$ that are also in the original structure
  $A$ from the elements that are added to $A^\dbstar$ by the transductions, we can
  rewrite~$\phi$ to a formula $\phi^\dbstar$, such that for each $A \in \CC$ we
  have $A \models \phi$ if, and only if, $A^\dbstar \models \phi^\dbstar$. In
  particular, $\phi^\dbstar$ is still an order-invariant \MSO-formula.

  In order to test whether $A^\dbstar \models \phi^\dbstar$ holds, we view
  $\phi^\dbstar$ as an $\MSO[\tau^\dbstar \cup \{\leq\}]$-formula and test whether
  $(A^\dbstar,\leq) \models \phi^\dbstar$ holds for some total order $\leq$ over
  $U(A^\dbstar)$ compatible with $A^\dbstar$. Using the terminology developed in
  Section~\ref{sec:inv-composition}, we ask whether $\phi^\dbstar$ is equivalent
  to a formula from a realizable type $\theta$ of $A^\dbstar$. Due to the
  order-invariance of $\phi^\dbstar$, this is equivalent to asking whether each
  realizable type $\theta$ contains a formula equivalent to $\phi^\dbstar$.  In
  order to have access to a realizable type of $A^\dbstar$, we define a compatible
  set of types $\Theta_r'$ for the root $r$ by using a \CMSO-formula that implements
  the following three parts: (1) It existentially guesses a cover
  $(P_1,\dots,P_m)$ of all nodes of the tree decomposition that \emph{induces} the
  set of types $\Theta'_t := \{\theta_i \mid i \in [m] \text{ with } t \in P_i\}$
  at each node $t$ of the tree decomposition. (2) It tests whether the induced
  set of types for each leaf is compatible. This is possible since leaves are
  always b-nodes and the substructures induced by their bags contain total
  orderings. (3) It compares the induced set of types of each inner node $t$ with
  the set of types that we get by applying Lemmas~\ref{lem:oibtype} (in the case
  of a b-node) or~\ref{lem:oiatype} (in the case of an a-nodes) to the cover $(P_1
  \cap N_+(t),\dots,P_m \cap N_+(t))$ of its children $N_+(t)$.

  Finally, we test whether $\phi^\dbstar$ is equivalent to a formula from a type
  $\theta \in \Theta_r'$. Overall, this results in a \CMSO-formula $\psi^\dbstar$ that is
  equivalent to $\phi^\dbstar$ on $A^\dbstar$. Since $\phi^\dbstar$ on $A^\dbstar$
  is constructed to be equivalent to $\phi$ on $A$ and \CMSO-transductions
  preserve \CMSO-definability, we know that there exists a \CMSO-formula $\psi$ on
  $\tau$ that is equivalent to $\phi$ on all structures from $\CC$.
\end{proof}

\begin{proof}[Proof of Theorem~\ref{th:lift-oifo}] 
  The arguments are the same as in the proof of Theorem~\ref{th:lift-oimso},
  except that we need to avoid the use of \CMSO-for\-mulas. First of all, this is
  possible for the initial transduction that produces the otx $A^\star$ from $A$
  since the theorem only talks about \MSO-definable ordered tree decompositions,
  not \CMSO-definable ones. Second, we need to avoid the use of \CMSO-formulas in
  the order-invariant compositions for a-nodes. During the proof of
  Lemma~\ref{lem:oiatype}, we translate an \OIMSO-formula on colored sets into an
  equivalent \CMSO-formula. If we start with an \OIFO-formula instead, then we are
  able to translate it into an equivalent \MSO-formula at this point in the
  proof. This follows from the fact that \FO\ has the same expressive power as
  \OIFO\ on this class of structures~\cite{BenediktS2009}. The resulting proof
  of Theorem~\ref{th:lift-oifo} produces an \MSO-formula instead of a
  \CMSO-formula. 
\end{proof}

\section{Defining Decompositions}
\label{sec:decompositions}

During the course of the present section, we use \MSO-transductions to extend
graphs with tree decompositions for them. The first transduction (developed in
Section~\ref{sec:transduction-into-atoms}) is used to prove
Theorems~\ref{th:tw-oimso} and \ref{th:tw-oifo}, which apply to graphs of
bounded tree width. The second transduction (reviewed in
Section~\ref{sec:transduction-into-3cc}) is used to prove
Theorems~\ref{th:minor-oimso} and~\ref{th:minor-oifo}, which apply to graphs
that exclude $K_{3,\ell}$ for some $\ell \in \NN$ as a minor. The present
section's results work with graphs instead of general structures. Thus, we set
$\tau = \{E\}$ throughout the section where $E$ is the (binary) edge relation
symbol.

The structures defined by the transductions are over the vocabulary $\tau^+ :=
\tau \cup \{V_S,V_T,E_T,R_\beta\}$ where $V_S$ and $V_T$ are unary, and $E_T$
and $R_\beta$ are binary. A \emph{tree extension} (\emph{tx} for short) of a
graph $G = (V,E)$ is a $\tau^+$-structure $G^+$ that extends $G$ by a tree
decomposition $(T,\beta)$ of $G$. Tree decompositions are encoded as part of txs
just like they are encoded as part of otxs in Section~\ref{sec:lift}, but
without including a partial order. The below transductions turn graphs of a
certain kind into tree extensions of a certain kind. In order to state the
results concisely, we use the following terminology: whenever we talk about the
bags and separators of a tree extension $G^+$, we refer to the bags and
separators, respectively, of the tree decomposition $(T,\beta)$ encoded by
$G^+$. For a class $\mathcal{C}$ of graphs and a class $\mathcal{D}$ of tree
extensions, we say that an $\MSO[\tau,\tau^+]$-transduction $\Lambda$
\emph{defines tree extensions from $\mathcal{D}$ for graphs from $\mathcal{C}$}
if the following holds for every $G \in \mathcal{C}$: we have $\emptyset
\subsetneq \Lambda[G] \subseteq \mathcal{D}$ and every $G^+ \in \Lambda[G]$ is
isomorphic to a tree extension of~$G$.

\subsection{Defining Tree Decompositions into Graphs without Clique Separators}
\label{sec:transduction-into-atoms}

A \emph{clique separator} in a graph $G$ is a set $S \subseteq V(G)$, such that
$G[S]$ is a clique (that means, there is an edge in $G$ between every pair of
vertices from $S$) and there are two vertices $v,w \in V(G) \setminus S$ that
are disconnected in $G \setminus S$. In this case, $S$ \emph{separates} $v$ and
$w$. 
An \emph{atom} is a graph without clique separators; in particular, atoms are
connected graphs. We prove the following lemma. 

\begin{lemma}
  \label{le:main-decomposition}
  Let $k \in \NN$. There is an $\MSO[\tau,\tau^+]$-transduction
  $\Lambda_{\tw\leq k}$ that defines tree extensions 
  for graphs of tree width at most $k$ where (1) the bags induce subgraphs that
  are atoms, and (2) the separators of the tree decompositions are cliques.
\end{lemma}

Our proof uses the graph-theoretic ideas behind a logspace
algorithm~\cite{ElberfeldS2016} for constructing tree decomposition of the kind
described by Lemma~\ref{le:main-decomposition} and shows how to define the
construction using an \MSO-transduction. The mentioned algorithm first
constructs decompositions along small clique separators of the graph and, then, refines the
decompositions by also taking larger clique separators into account. Since
graphs of tree width at most~$k$ only contain cliques of size at most $k+1$,
applying $k+1$ refinement steps turns a given graph of tree width at most $k$
into a tree decomposition that proves the lemma.

Formally, constructing tree decompositions via clique separators of a growing
size involves working with a refined notion of atoms. For $c \in \NN$, a
\emph{$c$-clique separator} is a clique separator of size at most $c$ and a
\emph{$c$-atom} is a graph that does not contain clique separators of size at
most~$c$. Like atoms, $c$-atoms are connected by definition.

For a graph $G$ and a constant $c \in \NN$, we build a graph $T_c$ where $V(T_c)$
consists of all maximal subgraphs of $G$ that are $c$-atoms, which are called
\emph{atom nodes}, and all $c$-clique separators, which are called
\emph{separator nodes}. In addition, to each $t \in V(T_c)$ we assign a bag
$\beta_c(t) \subseteq V(G)$ as follows: if $t$ is an atom node, then $\beta_c(t)$
is the vertex set of the corresponding atom, and if $t$ is a separator node,
then $\beta_c(t)$ is the corresponding separator. An edge is inserted between
every atom node $t$ and separator node $u$ with $\beta_c(u) \subseteq
\beta_c(t)$. While $T_c$ is not a tree in general, \cite{ElberfeldS2016} proved
that, if $G$ is a $(c-1)$-atom for $c \geq 1$, then $(T_c,\beta_c)$ is a tree
decomposition for $G$. 

\begin{fact}
	\label{fa:decomposition-step}
  Let $c \geq 1$ and $G$ be a $(c-1)$-atom. Then $(T_c,\beta_c)$ is a tree
  decomposition for $G$. Moreover,
  \begin{enumerate}
  \item atom nodes are only connected to separator nodes and vice versa, and
  \item all leaves are atom nodes.
  \end{enumerate}
\end{fact}

The previous fact provides us with a single step in the decomposition refinement
procedure outlined above. We apply it in order to move from tree decompositions
whose bags induce $(c-1)$-atoms to tree decompositions whose bags induce
$c$-atoms. This is similar to the approach of \cite{ElberfeldS2016}, which is
based on the following construction.

Let $G$ be a graph and $D_{(c-1)} = (T_{(c-1)},\beta_{(c-1)})$ a tree
decomposition of $G$, such that for each node $t \in V(T_{(c-1)})$ the bag
$\beta(t)$ induces a maximal subgraphs of $G$ that is a $(c-1)$-atom or it
induces a $(c-1)$-separator. Moreover, the tree decomposition satisfies the two
properties stated in Fact~\ref{fa:decomposition-step}: neighbors of atom nodes
are only separator nodes and vice versa, and all leaves are atom nodes. We
modify the tree decomposition into a decomposition $D_c$, such that it still
satisfies the same properties, except that the constant $c-1$ is replaced by
$c$. For each atom node $V(t)$, we consider the tree decomposition $D_c^t =
(T^t_c,\beta^t_c)$ of the $(c-1)$-atom $G[\beta(t)])$ that we get from applying
Fact~\ref{fa:decomposition-step}. We replace $t$ with $D_c^t$ inside $D_{(c-1)}$
as follows: if $t$ is the root of $T^t_c$, we just replace it with $D_c^t$.
If $t$ is not the root, it has a unique parent separator node $u$ and, in turn,
$u$ has a unique parent atom node $v$. We replace $t$ with $D_c^t$ and connect $u$
to the root of $D_c^t$, which is constructed as an atom node whose bag contains all
of $\beta(u)$. Similarly, $v$ is replaced with $D_c^v$ and the edge between $v$
and $u$ is redirected such that there is an edge to $u$ from the highest atom node
in $D_c^v$ (with respect to the root of $D_c^v$) that contains all of $\beta(u)$,
which is unique.
The following fact follows from \cite{ElberfeldS2016}. The arguments
about the shape of the decomposition directly follow from the construction.

\begin{fact}
  \label{fa:refinement-step}
  Let $G$, $D_{(c-1)}$, and the constructed $D_c$ be defined as in the previous
  paragraph. Then $D_c$ is a tree decomposition for $G$. Moreover, in $D_c$,
  \begin{enumerate}
  \item atom nodes are only connected to separator nodes and vice versa, and
  \item all leaves are atom nodes.
  \end{enumerate}    
\end{fact}

The final proof of Lemma~\ref{le:main-decomposition} shows how the construction
of Fact~\ref{fa:refinement-step} can be done by an \MSO-trans\-duc\-tion. It
also needs to turn a given graph, which can possibly be disconnected, into a
tree decomposition whose bags induce the connected components of the
graph. Since this is a special case that is not covered by the above
constructions, we first prove it separately. In the context of \MSO-definable
tree decompositions, we use the concept of tree extensions. In order to do that,
we use the following convention: when we say that the bags of a tree
decomposition (or tree extension) are $c$-atoms, we mean that the subgraphs
induced by the bags are $c$-atoms. We frequently use the fact that there is an
\MSO-formula for each of the following properties of vertex subsets $V'
\subseteq V$ of a given graph $G$: $V'$ is a clique separator, $V'$ is a
$c$-clique separator for some fixed, but arbitrary, $c \in \NN$, $G[V']$ is an
atom, $G[V']$ is a $c$-atom for some fixed, but arbitrary, $c \in \NN$.

\begin{lemma} 
  \label{le:0-atoms}
  There is an $\MSO[\tau,\tau^+]$-transduction $\Lambda_{\textnormal{comp}}$
  that defines for every graph $G$ a tree extension whose tree decomposition 
  \begin{enumerate}
  \item has a single node with an empty bag (representing the empty separator), and
  \item for each component of $G$ exactly one node whose bag equals
    the vertex set of it.
  \end{enumerate}
\end{lemma}
\begin{proof}
  The main idea is to guess, via parameters, a set of vertices of the graph
  whose copies in the tree extension represent the atoms and separators in the
  decomposition; the term \emph{represent} hints to the fact that we are able to define
  the vertex set of the corresponding atom or separator in an \MSO-definable way
  from the atom node or separator node, respectively. The transduction
  $\Lambda_{\textnormal{comp}}$ has three parameters $\rootparam_0$,
  $\atomparam_0$, and $\cliqparam_0$ and three levels: Level 1 contains copies of
  the vertices of the original graph $G$, level 2 contains the atom nodes of the
  decomposition, and level 3 contains the separator nodes of the
  decomposition. 

  First of all, the formula $\lambda_\textsc{valid}$ tests
  whether the parameters are chosen in a way that allows the other formulas to
  define the tree extension from them. It ensures the following properties:
  $\rootparam_0$ contains exactly one vertex that we call $v_r$ in the
  following, $\rootparam_0 \cup \atomparam_0$ contain exactly one vertex
  from each connected component of $G$, and $\cliqparam_0$ contains exactly
  one vertex that we call $v_c$ with $v_c \in \atomparam_0$. Thus, $v_c$ is used
  to both represent an atom and to represent the unique separator in the
  construction. 

  For each $v \in \rootparam_0 \cup \atomparam_0$, the transduction defines
  $\beta(v,2)$ to be the vertex set of the connected component in which $v$
  lies. Moreover, we set $\beta(v_c,3) := \emptyset$. We create an edge between
  $(v_c,3)$ and each $(v,2)$ for $v \in \atomparam_0$. Moreover, edges are
  oriented away from the root $v_r$.
\end{proof}

We are now ready to prove Lemma~\ref{le:main-decomposition}. The remaining
difficulty for the proof lies in defining the construction of
Fact~\ref{fa:refinement-step} in an \MSO-definable way, which involves defining
the construction of Fact~\ref{fa:decomposition-step} simultaneously for all atom
nodes. 

\begin{proof}[Proof of Lemma~\ref{le:main-decomposition}]  
  We first turn a given graph into a tree decomposition whose bags are the
  graph's connected components using $\Lambda_{\textnormal{comp}}$ from
  Lemma~\ref{le:0-atoms}. Next, we refine this decomposition $k+1$ times using
  \MSO-transductions that implement the construction from
  Fact~\ref{fa:refinement-step}. Finally, the lemma follows since
  \MSO-transductions are closed under composition.

  Let $c \geq 1$ and $G^+$ be a tree extension with a tree decomposition
  $(T,\beta) = (T_{(c-1)},\beta_{(c-1)})$ as described above. In order to \MSO-define the
  construction of Fact~\ref{fa:refinement-step}, we use an
  $\MSO[\tau^+,\tau^+]$-transduction $\Lambda_c$, which transforms tree extensions
  into tree extensions. Similar to the transduction of the proof of
  Lemma~\ref{le:0-atoms}, it has three parameters, but this time they are called
  $\rootparam_c$, $\atomparam_c$, and $\cliqparam_c$. Moreover, it has three
  levels: to level 1 we copy the vertex set of the underlying graph and
  decomposition nodes whose bags are not refined, level 2 contains newly
  constructed atom nodes, and level 3 contains newly constructed separator nodes.
  The parameters have to satisfy certain properties similar to the ones in the
  proof of Lemma~\ref{le:0-atoms}, but they are more involved due to the
  following reasons. First, we need to
  make sure that all atoms can be refined simultaneously. Second, we need to
  make sure that each new atom node represents a unique atom. In the proof of
  Lemma~\ref{le:0-atoms} the connected components, which are 0-atoms, are
  disjoint and, thus, it was possible to choose a vertex from each component. In
  the case of $c$-atoms for $c \geq 1$, a vertex can be part of multiple
  atoms. In order to work around this problem, we utilize the tree-like partial
  order that is given by the decomposition with respect to the chosen root.

  We start with the existing tree extension $G^+$ and consider where it needs to
  be modified. Since $\Lambda_c[G^+]$ will be a refinement of $G^+$ where new
  separator nodes are added, but existing separator nodes do not change, all of
  the separator nodes present in $G^+$ can be copied to level 1 directly without
  modification. On the other hand, the atom nodes in $G^+$ are refined if they
  contain a $c$-clique separator, so altogether the formula
  $\lambda^1_{V_T}(t)$ is satisfied only for some $t \in V(T)$: either if $t$
  is a separator node, that means, where $\beta(t)$ induces a clique of size
  up to $c$; or otherwise if the size of $\beta(t)$ is larger than $c$ and there
  is no $c$-clique separator. This effectively removes exactly those atom nodes
  which have a
  $c$-clique separator and which we thus need to decompose further.  We
  define $\lambda^1_{R_\beta}$, such that $R_\beta(t) = \beta(t)$ because we do
  not want the bags of these copied nodes to change, and similarly, the edges
  between any pair of copied nodes $s,t$ remain the same, so we define
  $\lambda^{1,1}_{E_T}(s,t)$ to be satisfied precisely if $(s,t) \in E(T)$.
  
  As a reminder, the indices of the formulas in a transduction specify a level
  for each of its free variables -- so as an example for a binary relation like
  $E_T$, the formula $\lambda^{2,3}_{E_T}(v,w)$ being satisfied for two concrete
  vertices $a = v$ and $b = w$ would mean that the vertex $(a,2)$ (the copy of $a$ on
  level 2) is connected to the vertex $(b,3)$ (the copy of $b$ on level 3) in the
  tree $T$ defined by the transduction.
  The transduction then constructs the relation $E_T$ by taking the union over all
  satisfying assignments of $\lambda^{i,j}_{E_T}(v,w)$ for all pairs of levels $i,j$.

  Next, we define the new atom and separator nodes, as well as their connectivity
  to the forest $D$ resulting from the described removal of atom nodes and their
  incident edges from $T$. Let $t \in V(T)$ be an atom node that is deleted and
  set $A_t \definedas G[\beta(t)]$, which contains at least one $c$-clique
  separator. We define a partial tree decomposition $D_t$ of $A_t$ into $c$-atoms,
  and then show how $D_t$ is reinserted into the forest $D$ in place of the
  deleted node $t$. We keep in mind that $A_t$ is a $(c-1)$-atom and, thus, free
  of any clique separators up to size $c-1$. Like in the construction of
  Fact~\ref{fa:refinement-step}, the root atom node of each decomposition $D_t$ is
  chosen so that it contains the $C = \beta(s)$ where $s$ is the parent
  separator node of $t$ in $T$. If $t$ is itself the root of $T$ and thus has no
  parent, then consider $C = \emptyset$ in the following.

  \textbf{Parameters of the transduction and their validity properties.}  We
  describe the properties of the parameters verified by
  $\lambda_\textsc{valid}$. They are used to single out a unique vertex of $G^+$
  for each $c$-atom and each $c$-clique separator, as well as a unique $c$-atom
  assigned to the root of each partial tree decomposition $D_t$ with the property
  describe above.

  The parameter $\rootparam_c$ contains exactly one vertex for each $c$-atom
  that includes a $c$-clique separator. We aim to find some $r \in A_t \cap
  \rootparam_c$ as the unique root vertex of a $c$-atom $A_t$. Since $r$ is
  supposed to represent the root node $(r,2)$ of $T_t$ that is later connected to
  the parent separator node $s$ of $t$ when reinserting the partial tree
  decomposition $D_t$ into the forest $D$, the root atom has to contain the
  clique $C$ in its own bag. There are
  potentially multiple atoms that satisfy this property.  If we consider the tree
  decomposition from Fact~\ref{fa:decomposition-step} on the subgraph $A_t$, then
  the set of $c$-atoms containing the clique $C$ form a subtree (due to the
  cover and connectedness condition of any tree decomposition).
  The leaves of this subtree are $c$-atoms which contain at
  least one vertex $r$ that is not present in any of the other $c$-atoms from
  $A_t$ that include all of $C$. Note that this either immediately implies that
  either $r \notin C$, or there is just a single candidate $c$-atom, in which case
  we may freely pick an $r \notin C$. This suffices as a unique identifier of the
  root $c$-atom of $A_t$, because then any other $c$-atom containing $r$ cannot
  contain all of $C$. An \MSO-formula can ensure that for each $(c-1)$-atom $A_t$ that
  is decomposed further, $\rootparam_c$ contains a single root vertex $r$ from the
  described candidates for this $A_t$. Since $r \notin C$, the respective $r$ can
  only appear in bags in the subtree of $T$ below $t$. So if $r$ appeared again in
  the root of a different $(c-1)$-atom that gets decomposed further, it would
  necessarily be in the bag of the separator node just above that $(c-1)$-atom, a
  contradiction. So there is a one-to-one correspondence between $(c-1)$-atoms $A_t$
  that get decomposed further and the vertices $r \in \rootparam_c$. This shows
  $\rootparam_c$ has the desired properties for all $A_t$ simultaneously on all of
  $G$.

  For the other $c$-atom representatives we utilize the fact that the
  $c$-cliques between any two $c$-atoms in $A_t$ can be linearly ordered. To see
  this, remember the construction of the tree decomposition from
  Fact~\ref{fa:decomposition-step}. In particular, for any vertex $v \in A_t$
  outside of the root $c$-atom, we can define the $c$-clique separator $S$
  \emph{closest to $v$ compared to the root atom} in the sense that $S$ separates
  a vertex of the root atom of $A_t$ from $v$, but no other $c$-clique $S'$
  separates a vertex of the root atom from both $v$ and a vertex of $S$. We define
  an \MSO-formula $\formel{closest-clique-separator}^c(v,S)$, which is satisfied
  exactly for vertices $v$ and $c$-cliques $S$ that satisfy this property. Note
  that this formula works globally on all of $G^+$, because the root vertex $r$ of
  each $A_t$ an be retrieved from the parameter $\rootparam_c$. So for each
  $c$-atom $A$ within $A_t$, we define the (nonempty) set $Z_{A}$ of vertices of
  this atom which are not in its closest $c$-clique separator. For different
  $c$-atoms, these sets are distinct -- since an overlap would mean that this
  vertex would appear in the $c$-clique~separator between them, which is then a
  closer clique~separator for one of the $c$-atoms, a contradiction. Via the
  parameter $\atomparam_c$, we guess a single vertex of $Z_{A}$ for each
  $c$-atom. An \MSO-formula can verify that conversely, no two vertices of
  $\atomparam_c$ are in the same set $Z_{A}$. This establishes the one-to-one
  correspondence of every $a \in \atomparam_c$ to the sets $Z_{A}$ and thus, the
  non-root $c$-atoms $A$ in all of $G$. Remember that the root
  atoms are already covered above by $\rootparam_c$.

  To define representatives for the separator nodes of $T_t$, we make the following
  observation: in the tree decomposition, each separator node will have at least
  one atom node as its child. Consequently, we use the representative of a child
  atom nodes also as the representative of its closest $c$-clique separator
  towards the root vertex $r$.  This overlap explains why we use separate levels
  for atom and separator nodes.  We use the parameter $\cliqparam_c$ to guess these
  representatives, and have to only verify $\cliqparam_c \subseteq \atomparam_c$,
  that no two vertices in $\cliqparam_c$ have the same closest~$c$-clique
  separator, and that for each $c$-clique separator $S$ some vertex $v \in
  \cliqparam_c$ exists that has $S$ as its closest separator. This guarantees the
  one-to-one correspondence of $c$-clique separators and vertices in
  $\cliqparam_c$ not just for $A_t$, but for all of~$G$.

  \textbf{Defining the construction of Fact~\ref{fa:decomposition-step}.}  We
  follow the construction of the decompositions from
  Fact~\ref{fa:decomposition-step}. We define the formula $\lambda^2_{V_T}(v)$
  such that it is satisfied exactly for the vertices $v \in \rootparam_c \cup
  \atomparam_c$, and $\lambda^3_{V_T}(v)$ such that it is defined exactly for the
  vertices $v \in \cliqparam_c$. The properties of these parameters as discussed
  above can be defined in \MSO.

  We now know that for a separator node $(v,3)$ created in this way, the
  clique separator it represents is the closest $c$-clique separator $S$ towards
  the unique root $r \in \rootparam_c \cap A_t$, which we can extract using the
  formula \MSO- formula $\formel{closest-clique-separator}^c(v,S)$ and thus set
  $R_\beta(v,3) = S$ by defining the formula $\lambda^{3,1}(v,x)$ such that it
  is satisfied exactly for $x \in S$. Conversely, for an atom node $(v,2)$ created
  this way, we extract the $c$-atom $A$ it represents by finding the closest
  $c$-clique separator $S$ towards $r$. The atom $A$ is then the set of vertices
  which either have $S$ as its closest $c$-clique separator, which means they are
  in the set $Z_{A}$ defined above, or which are itself in $S$. We can then
  define $R_\beta(v,3) = A$ similarly to above. This sets up the bags of the
  separator and atom nodes to be exactly the set of vertices of the clique separator
  (and respectively, atom) which they represent.

  Finally, we define the edges between nodes in $T_t$. Remember that the
  construction from Fact~\ref{fa:decomposition-step} connects an atom and a
  separator node if the bag of the separator node is completely contained in the bag
  of the atom node. We only have to define this as a directed tree decomposition
  rooted in the atom node $(r,2)$ of $A_t$ for the unique root
  representative $r \in A_t \cap \rootparam_c$.

  We use the formula $\lambda^{3,2}_{E_T}$ to express that there is an edge from
  a separator node $(u,3)$ to an atom node $(v,2)$ precisely if the vertices $u$ and
  $v$ have the same closest $c$-clique~separator $S$, which is unique.  Then by
  the above, the bag of $(u,3)$ is precisely this separator $S$ and this $S$ is
  completely contained in the bag of $(v,2)$, so the desired property is
  satisfied. Similarly, we use the formula $\lambda^{2,3}_{E_T}$ to express that
  there is an edge from an atom node $(v,2)$ to a separator node $(u,3)$ precisely if
  the bag of $(u,3)$ is completely contained in the bag of $(v,2)$, but they do
  not have the same closest $c$-clique~separator $S$. Thus, we constructed a
  tree decomposition $D_t$ as described in Fact~\ref{fa:decomposition-step}.

  \textbf{Defining the construction of Fact~\ref{fa:refinement-step}.}
  We now move from the view of the single $(c-1)$-atom $A_t$ and its tree
  decomposition $T_t$ to the global view on all of $G^+$. If we stopped defining
  the rest of the transduction here, the decomposition graph would now be the
  forest $D$ together with all partial tree decompositions $D_t$ for removed
  $(c-1)$-atom~nodes~$t$.  It remains to define how this forest is merged back
  together into a single tree decomposition.

  Let $t$ be a deleted $(c-1)$-atom node and $T_t$ the newly constructed tree of 
  the partial decomposition $D_t$ into $c$-atoms on the bag $A_t \definedas 
  G[\beta(t)]$. Further let $s$ be the parent of $t$ in $T$, which is a
  separator node. We use the formula $\lambda^{1,2}_{E_T}$ to define the edges from
  $s$ to the root of $T_t$ and thereby reattach the partial tree decompositions at
  the appropriate position. So $\lambda^{1,2}_{E_T}(s,u)$ is satisfied if $s$ is a
  separator node with a deleted child node, $u \in \rootparam_c$, and $\beta(s)$
  is the closest $c$-clique separator of $u$ since this means precisely that 
  the node $(u,2)$ has the root atom of $D_t$ as its bag.

  For the formulas $\lambda^{2,1}_{E_T}$, finding the correct point of
  attachment is a bit more involved. If $t$ had no child nodes,
  there is nothing to reattach. Otherwise we have to consider all former child
  nodes $s_1,\ldots,s_n$ of the deleted node $t$, each of which is a separator node
  according to Fact~\ref{fa:decomposition-step}. Each of their bags
  is a clique, and we would thus receive a valid tree decomposition if we
  connected each $s_j$ to an atom node $t_j$ of $T_t$ such that
  $\beta(s_j) \subseteq \beta(t_j)$ for all $j \in [n]$. Following the
  construction of Fact~\ref{fa:refinement-step}, to find a unique
  connection point, we take a closer look at the potential choices of the
  compatible atom nodes for a node $s_j$: due to the connectedness property, the
  set of nodes in $T_t$ which contain all of $\beta(s_j)$ is connected. This means
  that there is a unique atom node $t_j^*$ in this tree which lies closest to the
  root of $T_t$.  Moreover, because the set of nodes that include the clique
  $\beta(s_j)$ is connected in $T_t$, this node can be found in \MSO\ by asking
  for a node whose bag includes $\beta(s_j)$, but whose parent node in $T_t$ does
  not include $\beta(s_j)$.  We can thus define $\lambda^{2,1}_{E_T}(t,s)$ to be
  satisfied precisely if $s=s_j$ and $t=t_j^*$ hold, which is \MSO-definable.
  This concludes the reintegration of $T_t$ and finishes the description of the
  \MSO-transduction that implements the construction of
  Fact~\ref{fa:refinement-step}.
\end{proof}

\subsection{Defining Tree Decompositions into 3-Connected Components}
\label{sec:transduction-into-3cc}

A graph $G$ is \emph{$k$-connected} if $|G| > k$ and $G$ has no separator $S
\subseteq V(G)$ of size $|S| < k$. Courcelle \cite{Courcelle1999} showed that
one can use \MSO-transductions to define tree decompositions into 3-connected
components. We formulate this result with respect to the notion of tree
extensions as Fact~\ref{pro:3connected-mso-definable}.

\begin{fact}
\label{pro:3connected-mso-definable}
  There is an \MSO-transduction $\Lambda_{\textnormal{3-comp}}$ that defines 
  tree extensions whose torsos (1) are 3-connected, cycles, a single edge, or a
  single vertex, and (2) separators
  have size at most 2 for all graphs.
\end{fact}

The torsos of the tree decomposition produced by
Fact~\ref{pro:3connected-mso-definable} always induce topological subgraphs; a
\emph{topological subgraph} $G'$ of a graph $G$ arises by taking a subgraph of
$G$ and replacing some paths with edges. Later we use this insight since
whenever a graph $G$ does not contain a certain graph $H$ as a minor, then this
also holds for each of its topological subgraphs. In our application $H$ equals
$K_{3,\ell}$ for some $\ell \in \NN$.

\section{Defining Orderings}  
\label{sec:orderings}

In the previous section, we have seen how to define tree decompositions along
clique separators and discussed how to define tree decompositions into
3-connected components. In the present section we further define total orders
for the bags of these decompositions whenever our graphs have bounded tree width
or exclude a $K_{3,\ell}$-minor for some $\ell$. The latter covers planar graphs
since they exclude the minor $K_{3,3}$.

\subsection{Orderings Definable in Monadic Second-Order Logic}
\label{sec:orderings-mso}

Our bag orderings are based on applying the following result of Blumensath and
Courcelle~\cite{BlumensathC2014}. In order to state it formally, we introduce
some terminology. Let $\tau$ be a vocabulary that does not contain the binary
relation symbol $\leq$. We say that an $\MSO[\tau,\tau \cup
\{\leq\}]$-transduction $\Lambda$ \emph{defines orderings} on a class $\CC$ of
$\tau$-structures if the following holds for every $A \in \CC$: $\Lambda(A) \neq
\emptyset$ and every $B \in \Lambda(A)$ is an expansion of $A$ with a binary
relation $\le^B$ that is a linear order of $U(B)$. A class $\CC$ of graphs has
the \emph{bounded separability} property if there is a function $s \colon \NN
\to \NN$, such that for all graphs $G \in \CC$ and vertex sets $S \subseteq
V(G)$, the number of components of $G\setminus S$ is bounded by $f(|S|)$. The
below fact refers to \GSO-logic on graphs; it is defined by taking \MSO-logic on
graphs and extend it with the ability to quantify over subsets of a graph's
edges~\cite{Graedeletal2002}.

\begin{fact}
  \label{fa:mso-orderings}
  Let $\CC$ be a class of graphs with bounded separability that excludes
  $K_{\ell,\ell}$ as a minor for some $\ell \in \NN$. There is a
  \GSO-transduction $\Lambda_{\textsc{order-sep}}$ that defines total orderings
  on $\CC$. 
\end{fact}

Since \GSO-logic collapses to \MSO-logic on every class of graphs that exclude a
fixed minor~\cite{Courcelle2003} (in fact, this applies to the more general
class of uniformly $k$-sparse graphs, but we do not need them for our proofs),
and neither bounded tree width graphs nor the $K_{3,\ell}$-minor-free graphs
contain all complete bipartite minors, the fact has the following corollary.

\begin{corollary}
  Let $\CC$ be a class of graphs with bounded separability that excludes
  $K_{\ell,\ell}$ as a minor for some $\ell \in \NN$. There is an
  \MSO-transduction $\Lambda_{\textsc{order-sep}}$ that defines total orderings
  on $\CC$.  
\end{corollary}

\subsection{Defining Orderings in the Bounded Tree Width Case}
\label{sec:ordering-tw}

In general, it is not possible to totally order atoms of bounded tree width in
\MSO\ or, even, \CMSO. An example being a graph made up by $n$ cycles of length
$n$ each connected to two universal vertices $u_1$ and $u_2$, but without an
edge between $u_1$ and $u_2$. Graphs of this kind have bounded tree width and
are atoms, but \CMSO\ is not able to define a total ordering on the graph's
vertices. In the following we show how to preprocess given graphs, such that the
resulting atoms cannot be of the above kind. In particular, the preprocessing
ensures that the two universal vertices in the above example have an edge
between them and, thus, the considered graph is no longer an atom.

Given a graph $G$, its \emph{improved version} $G'$ is the graph with vertex set
$V(G') := V(G)$ and $(v,w) \in E(G')$ holds for every two distinct vertices $v,w
\in V(G')$ if, and only if, $(v,w) \in E(G)$ or there are $\tw(G) + 1$
internally disjoint paths between $v$ and $w$ in $G$. Computing the improved
version of a graph is commonly part of algorithms that construct tree
decompositions~\cite{Lokshtanovetal2014}. Pairs of vertices with $\tw(G) + 1$
internally-disjoint paths between them always lie in a common bag in every tree
decomposition. Thus, connecting pairs with this property with an edge does not
change the tree decompositions of the graph and, moreover, it simplifies the
task of constructing tree decompositions by producing a graph that is closer to
embeddings into $k$-trees for $k = \tw(G)$ than the original graph. The
\MSO-transduction of the below proposition is based on defining a constant
number, $k+1$, of disjoint paths between pairs of vertices of the graph. This
can be done by using $k+1$ set variables where each set colors the vertices of a
single path that does not share vertices with other paths.

\begin{proposition}
  \label{pr:improve}
  Let $k \in \NN$. There is an \MSO-transduction $\Lambda_{\textsc{improve}}$
  that defines the improved version for every graph of tree width at most $k$.   
\end{proposition}

Since \MSO-transductions are closed under composition, we continue to work with
the improved version of the graph instead of the original input graph.

The main reason behind the non-definability of total orderings in the above
example lies in the fact that there is an unbounded number of subgraphs
connected to each other via a small separator. This is not possible if we look
at the bags of decomposed improved graphs. 

\begin{lemma}
  \label{le:tw-separability}
  Let $\CC$ be a class of graphs of bounded tree width that are
  improved and atoms. Then $\CC$ has the
  bounded separability property. 
\end{lemma}
\begin{proof}
  Let $G \in \CC$ and $k:=\tw(G)$. Let $S \subseteq V(G)$, and let
  $G_1,\dots,G_n$ be the components of $G\setminus S$.  We shall prove
  that $n\le  {|S| \choose 2} \cdot k + 1$ holds.

  Without loss of generality we assume that $n\ge 2$. For every
  $i\in[n]$, let $S_i$ be the set of neighbors of $G_i$ in $S$. As
  $G$ is an atom, $S_i$ is not a clique in $G$. Thus there are $u,v\in
  S_i$ such that $\{u,v\} \notin E(G)$. Since $G$ is
  improved, we have $u,v\in S_i$ for at most $k$ indices
  $i\in[n]$. As there are $\binom{|S|}{2}$ pairs $\{u,v\}\subseteq S$,
  this implies $n\le\binom{|S|}{2}k$ and, thus, the above inequality holds.  
\end{proof}

We get the following from combining Lemma~\ref{le:tw-separability}
with Fact~\ref{fa:mso-orderings}. 

\begin{corollary}
  \label{cor:ord-tw}
  Let $\CC$ be a class of graphs of bounded tree width that are
  improved and atoms. There is an \MSO-transduction
  $\Lambda_{\textsc{order-tw}}$ that defines a total ordering for
  every $G \in \CC$.
\end{corollary}

Using the definable decompositions from the previous section and the just
developed definable orderings, we can prove the results about bounded
tree width and \OIMSO\ as well as \OIFO.

\begin{theorem}
  \label{th:tw-oimso}
  Let $\CC$ be a class of graphs with bounded tree width. Then $\OIMSO = \CMSO$
  on $\CC$. 
\end{theorem}
\begin{proof} 
  We show that $\CC$ admits $\MSO$-definable (hence $\CMSO$-de\-fi\-na\-ble)
  ordered tree decompositions of bounded adhesion. This proves the theorem by
  applying Theorem~\ref{th:lift-oimso}, the lifting theorem for $\OIMSO$. Let $k$
  be a tree width bound for the graphs from $\CC$. Instead of directly working
  with the structure $A$, we work with its Gaifman graph $G' = G(A)$, which has
  the same tree decompositions and is \MSO-definable in $A$. We start to define
  the improved version $G'$ in $G$ using the \MSO-transduction
  $\Lambda_{\textsc{improve}}$ from Proposition~\ref{pr:improve}. Next, we apply
  the transduction $\Lambda$ of Lemma~\ref{le:main-decomposition} to $G'$, which
  defines a tree extension $G^+$. The bags of the tree decomposition underlying
  the tree extension induce subgraphs that are atoms, and all adhesion sets are
  cliques. Since $G$ and, hence, also $G'$ has tree width $k$ and graphs of tree
  width at most $k$ only contain cliques of size at most $k+1$, this implies a
  bounded adhesion (the adhesion is bounded by $k+1$). In order to obtain an otx,
  we need to add total orderings for each bag. The bags of the tree decomposition
  obtained so far induce atoms and, since $G'$ is an improved graph, these atoms
  are improved, too. That means, we can now use the transduction
  $\Lambda_{\textsc{order-tw}}$ from Corollary~\ref{cor:ord-tw} to obtain a total
  ordering for a given bag. In order to define orderings for all bags at the same
  time, we utilize the decomposition's bounded adhesion in the following
  way. Transduction $\Lambda_{\textsc{order-tw}}$ orders a single bag by using a
  collection of set parameters, which are vertex colorings from which we can
  define the ordering. If we now want to order different neighboring bags at the
  same time, these vertex colorings might interfere in a way that makes it
  impossible to reconstruct an ordering.
  
  We can do the following: as our (improved) graph has tree width at most $k$, it
  has coloring number at most $k+1$, and thus we can first guess a
  proper $(k+1)$-coloring where no two adjacent vertices have the same
  color. In particular, this implies that for
  each adhesion set $S$ that occurs, all elements of $S$ have different
  colors, because they are cliques. This gives us a way to
  simultaneously get a linear order of all adhesion sets by just
  fixing an order on the $(k+1)$ colors. Let us call the
  $(k+1)$-colors we used this way our \emph{adhesion colors}.
   
  Now we guess a collection of colors that we would like to
  use to order the bags at the atom nodes. (The bags at separator nodes
  are just adhesion sets and thus already ordered by the adhesion
  colors.) We globally guess a suitable collection of colors. Let us call them
  \emph{bag colors}. Within each bag $B$
  of the tree, we ignore the colors in the adhesion
  (upward) adhesion set $S$ and instead consider all extensions of the
  coloring of the remaining nodes that lead to a linear order of the
  bag. There is only a bounded number of such extensions, and as the
  adhesion set $S$ is linearly ordered, we can use the
  lexicographically smallest of these extensions to define the order.
\end{proof}

\begin{theorem}
  \label{th:tw-oifo}
  Let $\CC$ be a class of graphs with bounded tree width. Then $\OIFO \subseteq
  \MSO$ on $\CC$.   
\end{theorem}
\begin{proof}
  We use the proof of Theorem~\ref{th:tw-oimso}, but apply
  Theorem~\ref{th:lift-oifo}, the lifting theorem for $\OIFO$, instead of
  Theorem ~\ref{th:lift-oimso}, the lifting theorem for $\OIMSO$. 
\end{proof}

\subsection{Defining Orderings in the $K_{3,\ell}$-Minor-Free Case}
\label{sec:ordering-minor}

Like in the previous section, we want to apply Fact~\ref{fa:mso-orderings} to
define total orderings, but this time use it for graphs that are 3-connected and
do not contain $K_{3,\ell}$ as a minor for some $\ell \in \NN$.

\begin{lemma}
  \label{le:minor-separability}
  Let $\CC$ be a class of 3-connected graphs that exclude a
  $K_{3,\ell}$-minor for some $\ell \in \NN$. Then $\CC$ has the bounded
  separability property.  
\end{lemma}
\begin{proof}
  Let $G$ be a 3-connected graph that does not contain $K_{3,\ell}$ for some
  $\ell \in \NN$ as a minor and $S \subseteq V(G)$ with $k = |S|$. Now let
  $G_1,\dots,G_n$ be the components of $G\setminus S$. If $k \leq 2$, then $n \leq 1$
  since $G$ is 3-connected. If $k \geq 3$, 3-connectedness implies that every
  component is connected to at least 3 vertices in $S$. For the sake of
  contradiction, assume $n \geq \ell {k \choose 3}$. Then there exists a subset $T$
  of $S$ with $T = 3$ that is connected to at least $\ell$
  components. By deleting everything except $T$ and these components as well as
  contracting the components we produce the minor $K_{3,\ell}$. Since this is not
  possible, we have $n < \ell {k \choose 3}$ and hence bounded separability.   
\end{proof}

\begin{corollary}
  \label{co:minor-ordering}
  Let $\CC$ be a class of 3-connected graphs that exclude a $K_{3,\ell}$-minor
  for some $\ell \in \NN$. There is an \MSO-transduction
  $\Lambda_{\textsc{order-minor}}$ that defines a total ordering for every $G \in
  \CC$.   
\end{corollary}

Combining the decompositions from the previous section with the
ordering from Corollary~\ref{co:minor-ordering}, we can prove the following. 

\begin{theorem}
  \label{th:minor-oimso}
  Let $\CC$ be a class of graphs that exclude $K_{3,\ell}$ as a minor for some
  $\ell \in \NN$. Then $\OIMSO = \CMSO$ on $\CC$.
\end{theorem}
\begin{proof}
  The proof is similar to the proof of Theorem~\ref{th:tw-oimso}, except that we
  need to use different transductions to define the tree decomposition and the
  ordering for the bags. Everything else remains the same since we still work
  with tree decompositions that have a bounded adhesion (in this case, the
  maximum adhesion is 2) and apply the lifting theorem for \OIMSO. For
  constructing a tree decomposition of bounded adhesion, we use
  Fact~\ref{pro:3connected-mso-definable}. For constructing the bag orderings, we
  follow the arguments from Theorem~\ref{th:tw-oimso}, but apply
  Corollary~\ref{co:minor-ordering} to the torsos of the decomposition combined
  with the observation that graphs that exclude a minor can be properly colored
  with a bounded number of colors. 
\end{proof}

\begin{theorem}
  \label{th:minor-oifo}
  Let $\CC$ be a class of graphs that exclude $K_{3,\ell}$ as a minor for some
  $\ell \in \NN$. Then $\OIFO \subseteq \MSO$ on $\CC$.
\end{theorem}
\begin{proof}
  Similar to the idea in the proof of Theorem~\ref{th:tw-oifo}. We take the
  proof of Theorem~\ref{th:minor-oimso}, but use the lifting theorem for \OIFO\
  instead of the lifting theorem for \OIMSO. 
\end{proof}

\section{Conclusions}
\label{sec:conclusion}

We proved two lifting definability theorems, which show that if a class $\CC$ of
structures admits \MSO-defin\-able ordered tree extensions, then $\OIMSO = \CMSO$
and $\OIFO \subseteq \MSO$ on $\CC$. Using the lifting theorems in conjunction
with definable tree decompositions and definable bag orderings, we were able to
show that $\OIMSO = \CMSO$ and $\OIFO \subseteq \MSO$ hold for every class of
graphs (and structures) of bounded tree width and every class of graphs (and
structures) that exclude $K_{3,\ell}$ for some $\ell \in \NN$ as a minor. The
latter covers planar graphs.

Seeing the wide range of applications of the lifting theorems, it seems
promising to apply or extend them in order to handle every graph class defined
by excluding minors in future works. Moreover, an interesting question is
whether the $\OIFO \subseteq \MSO$ in Theorem~\ref{th:lift-oifo} can be turned
into an equality; possibly by using a logic more restrictive than
\MSO.

\section*{Acknowledgements}

We thank Pascal Schweitzer for the idea of
Lemma~\ref{le:tw-separability}.    

\bibliographystyle{abbrvurl}
\bibliography{main}

\end{document}